\documentclass[10pt,final,conference,oneside,twocolumn]{IEEEtran}
\pdfoutput=1

\usepackage{mathtools} 
\usepackage{amstext} 
\usepackage{amssymb} 
\usepackage{bbold} 
\usepackage{amsthm} 
\usepackage{stmaryrd} 

\usepackage[utf8]{inputenc}

\usepackage{relsize} 
\usepackage{microtype} 
\usepackage{csquotes} 
\usepackage{cancel}
\usepackage{pifont}

\usepackage{hyperref} 
\usepackage[compress]{cite} 

\usepackage{graphicx} 
\usepackage[usenames,dvipsnames]{xcolor} 
\usepackage{stackengine}
\usepackage{tikz} 
\usepackage{circuitikz} 
\usetikzlibrary{
	arrows,
	shapes,
	decorations,
	intersections,
	backgrounds,
	positioning,
	circuits.ee.IEC
	}
\usepackage{tikzit}
\usepackage{tikz-cd}

\usepackage{rotating}
\usepackage{floatpag}
\usepackage{xifthen}
\usepackage[normalem]{ulem}

\usepackage{dsfont}

		\newcounter{theorem_c} 

		\theoremstyle{plain} 
		\newtheorem{theorem}[theorem_c]{Theorem}
		
		\newtheorem{lemma}[theorem_c]{Lemma}

		\newtheoremstyle{exampstyle}
		  {2mm} 
		  {2mm} 
		  {\itshape} 
		  {} 
		  {\bfseries} 
		  {.} 
		  {.5em} 
		  {} 

		\theoremstyle{exampstyle}
		\newtheorem{definition}[theorem_c]{Definition}

		\newtheorem*{remark*}{Remark}
\newcommand{\ket}[1]{\vert #1 \rangle} 
\newcommand{\bra}[1]{\langle #1 \vert}

\newcommand{\id}[1]{{id}_{#1}}
\newcommand{\obj}[1]{\operatorname{obj}\left(#1\right)}
\newcommand{\mor}[1]{\operatorname{mor}\left(#1\right)}
\newcommand{\inputs}[1]{\operatorname{in}\!\left(#1\right)}
\newcommand{\outputs}[1]{\operatorname{out}\!\left(#1\right)}
\newcommand{\nodes}[1]{\operatorname{nodes}\!\left(#1\right)}
\newcommand{\edges}[1]{\operatorname{edges}\!\left(#1\right)}

\newcommand{\suchthat}[2]{\left\{ #1 \middle| #2\right\}}

\newcommand{\Mix}{\textbf{Mix}\xspace}
\newcommand{\DMix}{\textbf{DMix}\xspace}

\usepackage{xspace}

\newcommand{\dotonly}[1]{%
\,\begin{tikzpicture}[dotpic]
\node [#1] (a) at (0,0) {};
\end{tikzpicture}\,}

\newcommand{\graydot}{\dotonly{gray dot}\xspace}

\tikzstyle{env}=[copoint,regular polygon rotate=0,minimum width=0.2cm, fill=black]

\tikzstyle{probs}=[shape=semicircle,fill=white,draw=black,shape border rotate=180,minimum width=1.2cm]

\newcommand{\edgearrow}{{\arrow[black]{>}}}
\tikzstyle{every picture}=[baseline=-0.25em,scale=0.6]
\tikzstyle{dotpic}=[] 
\tikzstyle{math matrix}=[matrix of math nodes,left delimiter=(,right delimiter=),inner sep=2pt,column sep=1em,row sep=0.5em,nodes={inner sep=0pt},text height=1.5ex, text depth=0.25ex]

\tikzstyle{normall} = [every edge/.style={normal}]

\tikzstyle{smalln}=[scale=0.6]
\tikzstyle{small}=[scale=0.6, every node/.style={scale=0.6}]
\tikzstyle{smalllabels} = [every node/.style={scale=0.6}]

\tikzstyle{inline text}=[text height=1.5ex, text depth=0.25ex,yshift=0.5mm]
\tikzstyle{label}=[font=\footnotesize,text height=1.5ex, text depth=0.25ex,yshift=0.5mm]
\tikzstyle{left label}=[label,anchor=east,xshift=1.5mm]
\tikzstyle{right label}=[label,anchor=west,xshift=-1.5mm]
\tikzstyle{diredges}=[every to/.style={diredge}]

\tikzstyle{braceedge}=[decorate,decoration={brace,amplitude=2mm,raise=-1mm}]
\tikzstyle{small braceedge}=[decorate,decoration={brace,amplitude=1mm,raise=-1mm}]

\tikzstyle{normal}=[line width=0.9pt]
\tikzstyle{doubled}=[line width=1pt] 
\tikzstyle{boldedge}=[doubled,shorten <=-0.17mm,shorten >=-0.17mm]
\tikzstyle{boldedgegray}=[doubled,gray,shorten <=-0.17mm,shorten >=-0.17mm]
\tikzstyle{singleedgegray}=[gray]

\tikzstyle{semidoubled}=[line width=1.4pt] 
\tikzstyle{semiboldedgegray}=[semidoubled,gray,shorten <=-0.17mm,shorten >=-0.17mm]

\tikzstyle{boxedge}=[semiboldedgegray]

\tikzstyle{boldedgedashed}=[very thick,dashed,shorten <=-0.17mm,shorten >=-0.17mm]
\tikzstyle{vboldedgedashed}=[doubled,dashed,shorten <=-0.17mm,shorten >=-0.17mm]
\tikzstyle{left hook arrow}=[left hook-latex]
\tikzstyle{right hook arrow}=[right hook-latex]
\tikzstyle{sembracket}=[line width=0.5pt,shorten <=-0.07mm,shorten >=-0.07mm]

\tikzstyle{causal edge}=[->,thick,gray]
\tikzstyle{causal nondir}=[thick,gray]
\tikzstyle{timeline}=[thick,gray, dashed]

\tikzstyle{cedge}=[<->,thick,gray!70!white]

\tikzstyle{empty diagram}=[draw=gray!40!white,dashed,shape=rectangle,minimum width=1cm,minimum height=1cm]
\tikzstyle{empty diagram small}=[draw=gray!50!white,dashed,shape=rectangle,minimum width=0.6cm,minimum height=0.5cm]

\tikzstyle{dot}=[inner sep=0mm,minimum width=2mm,minimum height=2mm,draw,shape=circle]  
\tikzstyle{Wsquare}=[white dot, shape=regular polygon, rounded corners=0.8 mm, minimum size=3.3 mm, regular polygon sides=3, outer sep=-0.2mm]
\tikzstyle{Wsquareadj}=[white dot, shape=regular polygon, rounded corners=0.8 mm, minimum size=3.3 mm, regular polygon sides=3, outer sep=-0.2mm, regular polygon rotate=180]
\tikzstyle{ddot}=[inner sep=0mm, doubled, minimum width=2.5mm,minimum height=2.5mm,draw,shape=circle]

\tikzstyle{black dot}=[dot,fill=black]
\tikzstyle{white dot}=[dot,fill=white,,text depth=-0.2mm]
\tikzstyle{white Wsquare}=[Wsquare,fill=gray,,text depth=-0.2mm]
\tikzstyle{white Wsquareadj}=[Wsquareadj,fill=white,,text depth=-0.2mm]
\tikzstyle{green dot}=[white dot] 
\tikzstyle{gray dot}=[dot,fill=gray!40!white,,text depth=-0.2mm]
\tikzstyle{red dot}=[gray dot] 

\tikzstyle{black ddot}=[ddot,fill=black]
\tikzstyle{white ddot}=[ddot,fill=white]
\tikzstyle{gray ddot}=[ddot,fill=gray!40!white]

\tikzstyle{gray edge}=[gray!60!white]

\tikzstyle{small dot}=[inner sep=0.5mm,minimum width=0pt,minimum height=0pt,draw,shape=circle]

\tikzstyle{small black dot}=[small dot,fill=black]
\tikzstyle{small white dot}=[small dot,fill=white]
\tikzstyle{small gray dot}=[small dot,fill=gray!40!white]

\tikzstyle{causal dot}=[inner sep=0.4mm,minimum width=0pt,minimum height=0pt,draw=white,shape=circle,fill=gray!40!white]

\tikzstyle{phase dimensions}=[minimum size=5mm,font=\footnotesize,rectangle,rounded corners=2.5mm,inner sep=0.2mm,outer sep=-2mm]
\tikzstyle{dphase dimensions}=[minimum size=5mm,font=\footnotesize,rectangle,rounded corners=2.5mm,inner sep=0.2mm,outer sep=-2mm]

\tikzstyle{white phase dot}=[dot,fill=white,phase dimensions]
\tikzstyle{white phase ddot}=[ddot,fill=white,dphase dimensions]

\tikzstyle{white rect ddot}=[draw=black,fill=white,doubled,minimum size=5mm,font=\footnotesize,rectangle,rounded corners=2.5mm,inner sep=0.2mm]
\tikzstyle{gray rect ddot}=[draw=black,fill=gray!40!white,doubled,minimum size=6mm,font=\footnotesize,rectangle,rounded corners=3mm]

\tikzstyle{gray phase dot}=[dot,fill=gray!40!white,phase dimensions]
\tikzstyle{gray phase ddot}=[ddot,fill=gray!40!white,dphase dimensions]
\tikzstyle{grey phase dot}=[gray phase dot]
\tikzstyle{grey phase ddot}=[gray phase ddot]

\tikzstyle{small phase dimensions}=[minimum size=4mm,font=\tiny,rectangle,rounded corners=2mm,inner sep=0.2mm,outer sep=-2mm]
\tikzstyle{small dphase dimensions}=[minimum size=4mm,font=\tiny,rectangle,rounded corners=2mm,inner sep=0.2mm,outer sep=-2mm]

\tikzstyle{small gray phase dot}=[dot,fill=gray!40!white,small phase dimensions]
\tikzstyle{small gray phase ddot}=[ddot,fill=gray!40!white,small dphase dimensions]

\tikzstyle{small map}=[draw,shape=rectangle,minimum height=4mm,minimum width=4mm,fill=white]

\tikzstyle{cnot}=[fill=white,shape=circle,inner sep=-1.4pt]

\tikzstyle{asym hadamard}=[fill=white,draw,shape=NEbox,inner sep=0.6mm,font=\footnotesize,minimum height=4mm]
\tikzstyle{asym hadamard conj}=[fill=white,draw,shape=NWbox,inner sep=0.6mm,font=\footnotesize,minimum height=4mm]
\tikzstyle{asym hadamard dag}=[fill=white,draw,shape=SEbox,inner sep=0.6mm,font=\footnotesize,minimum height=4mm]

\tikzstyle{hadamard}=[fill=white,draw,inner sep=0.6mm,font=\footnotesize,minimum height=4mm,minimum width=4mm]
\tikzstyle{small hadamard}=[fill=white,draw,inner sep=0.6mm,minimum height=1.5mm,minimum width=1.5mm]
\tikzstyle{small hadamard rotate}=[small hadamard,rotate=45]
\tikzstyle{dhadamard}=[hadamard,doubled]
\tikzstyle{small dhadamard}=[small hadamard,doubled]
\tikzstyle{small dhadamard rotate}=[small hadamard rotate,doubled]
\tikzstyle{antipode}=[white dot,inner sep=0.3mm,font=\footnotesize]

\tikzstyle{scalar}=[diamond,draw,inner sep=0.5pt,font=\small]
\tikzstyle{dscalar}=[diamond,doubled, draw,inner sep=0.5pt,font=\small]

\tikzstyle{small box}=[rectangle,inline text,fill=white,draw,minimum height=5mm,yshift=-0.5mm,minimum width=5mm,font=\small]
\tikzstyle{small gray box}=[small box,fill=gray!30]
\tikzstyle{medium box}=[rectangle,inline text,fill=white,draw,minimum height=5mm,yshift=-0.5mm,minimum width=10mm,font=\small]
\tikzstyle{square box}=[small box] 
\tikzstyle{medium gray box}=[small box,fill=gray!30]
\tikzstyle{semilarge box}=[rectangle,inline text,fill=white,draw,minimum height=5mm,yshift=-0.5mm,minimum width=12.5mm,font=\small]
\tikzstyle{large box}=[rectangle,inline text,fill=white,draw,minimum height=5mm,yshift=-0.5mm,minimum width=15mm,font=\small]
\tikzstyle{large gray box}=[small box,fill=gray!30]

\tikzstyle{Bayes box}=[rectangle,fill=black,draw, minimum height=3mm, minimum width=3mm]

\tikzstyle{gray square point}=[small box,fill=gray!50]

\tikzstyle{dphase box white}=[dhadamard]
\tikzstyle{dphase box gray}=[dhadamard,fill=gray!50!white]
\tikzstyle{phase box white}=[hadamard]
\tikzstyle{phase box gray}=[hadamard,fill=gray!50!white]

\tikzstyle{point}=[regular polygon,regular polygon sides=3,draw,scale=0.75,inner sep=-0.5pt,minimum width=9mm,fill=white,regular polygon rotate=180]
\tikzstyle{point nosep}=[regular polygon,regular polygon sides=3,draw,scale=0.75,inner sep=-2pt,minimum width=9mm,fill=white,regular polygon rotate=180]
\tikzstyle{copoint}=[regular polygon,regular polygon sides=3,draw,scale=0.75,inner sep=-0.5pt,minimum width=9mm,fill=white]
\tikzstyle{dpoint}=[point,doubled]
\tikzstyle{dcopoint}=[copoint,doubled]

\tikzstyle{pointgrow}=[shape=cornerpoint,kpoint common,scale=0.75,inner sep=3pt]
\tikzstyle{pointgrow dag}=[shape=cornercopoint,kpoint common,scale=0.75,inner sep=3pt]

\tikzstyle{wide copoint}=[fill=white,draw,shape=isosceles triangle,shape border rotate=90,isosceles triangle stretches=true,inner sep=0pt,minimum width=1.5cm,minimum height=6.12mm]
\tikzstyle{wide point}=[fill=white,draw,shape=isosceles triangle,shape border rotate=-90,isosceles triangle stretches=true,inner sep=0pt,minimum width=1.5cm,minimum height=6.12mm,yshift=-0.0mm]
\tikzstyle{wide point plus}=[fill=white,draw,shape=isosceles triangle,shape border rotate=-90,isosceles triangle stretches=true,inner sep=0pt,minimum width=1.74cm,minimum height=7mm,yshift=-0.0mm]

\tikzstyle{wide dpoint}=[fill=white,doubled,draw,shape=isosceles triangle,shape border rotate=-90,isosceles triangle stretches=true,inner sep=0pt,minimum width=1.5cm,minimum height=6.12mm,yshift=-0.0mm]

\tikzstyle{tinypoint}=[regular polygon,regular polygon sides=3,draw,scale=0.55,inner sep=-0.15pt,minimum width=6mm,fill=white,regular polygon rotate=180] 

\tikzstyle{white point}=[point]
\tikzstyle{white dpoint}=[dpoint]
\tikzstyle{green point}=[white point] 
\tikzstyle{white copoint}=[copoint]
\tikzstyle{gray point}=[point,fill=gray!40!white]
\tikzstyle{gray dpoint}=[gray point,doubled]
\tikzstyle{red point}=[gray point] 
\tikzstyle{gray copoint}=[copoint,fill=gray!40!white]
\tikzstyle{gray dcopoint}=[gray copoint,doubled]

\tikzstyle{white point guide}=[regular polygon,regular polygon sides=3,font=\scriptsize,draw,scale=0.65,inner sep=-0.5pt,minimum width=9mm,fill=white,regular polygon rotate=180]

\tikzstyle{black point}=[point,fill=black,font=\color{white}]
\tikzstyle{black copoint}=[copoint,fill=black,font=\color{white}]

\tikzstyle{tiny gray point}=[tinypoint,fill=gray!40!white]

\tikzstyle{diredge}=[decoration={markings, mark=at position 0.55 with \edgearrow},postaction=decorate]
\tikzstyle{ddiredge}=[<->]
\tikzstyle{rdiredge}=[<-]
\tikzstyle{thickdiredge}=[->, very thick]
\tikzstyle{pointer edge}=[->,very thick,gray]
\tikzstyle{pointer edge part}=[very thick,gray]
\tikzstyle{dashed edge}=[dashed]
\tikzstyle{thick dashed edge}=[very thick,dashed]
\tikzstyle{thick gray dashed edge}=[thick dashed edge,gray!40]
\tikzstyle{thick map edge}=[very thick,|->]

\makeatletter
\newcommand{\boxshape}[3]{%
\pgfdeclareshape{#1}{
\inheritsavedanchors[from=rectangle] 
\inheritanchorborder[from=rectangle]
\inheritanchor[from=rectangle]{center}
\inheritanchor[from=rectangle]{north}
\inheritanchor[from=rectangle]{south}
\inheritanchor[from=rectangle]{west}
\inheritanchor[from=rectangle]{east}
\backgroundpath{
\southwest \pgf@xa=\pgf@x \pgf@ya=\pgf@y
\northeast \pgf@xb=\pgf@x \pgf@yb=\pgf@y

\@tempdima=#2
\@tempdimb=#3

\pgfpathmoveto{\pgfpoint{\pgf@xa - 5pt + \@tempdima}{\pgf@ya}}
\pgfpathlineto{\pgfpoint{\pgf@xa - 5pt - \@tempdima}{\pgf@yb}}
\pgfpathlineto{\pgfpoint{\pgf@xb + 5pt + \@tempdimb}{\pgf@yb}}
\pgfpathlineto{\pgfpoint{\pgf@xb + 5pt - \@tempdimb}{\pgf@ya}}
\pgfpathlineto{\pgfpoint{\pgf@xa - 5pt + \@tempdima}{\pgf@ya}}
\pgfpathclose
}
}}

\boxshape{NEbox}{0pt}{5pt}
\boxshape{SEbox}{0pt}{-5pt}
\boxshape{NWbox}{5pt}{0pt}
\boxshape{SWbox}{-5pt}{0pt}
\boxshape{EBox}{-3pt}{3pt}
\boxshape{WBox}{3pt}{-3pt}
\makeatother

\tikzstyle{cloud}=[shape=cloud,draw,minimum width=1.5cm,minimum height=1.5cm]

\tikzstyle{map}=[draw,shape=NEbox,inner sep=2pt,minimum height=6mm,fill=white]
\tikzstyle{dashedmap}=[draw,dashed,shape=NEbox,inner sep=2pt,minimum height=6mm,fill=white]
\tikzstyle{mapdag}=[draw,shape=SEbox,inner sep=2pt,minimum height=6mm,fill=white]
\tikzstyle{mapadj}=[draw,shape=SEbox,inner sep=2pt,minimum height=6mm,fill=white]
\tikzstyle{maptrans}=[draw,shape=SWbox,inner sep=2pt,minimum height=6mm,fill=white]
\tikzstyle{mapconj}=[draw,shape=NWbox,inner sep=2pt,minimum height=6mm,fill=white]

\tikzstyle{medium map}=[draw,shape=NEbox,inner sep=2pt,minimum height=6mm,fill=white,minimum width=7mm]
\tikzstyle{medium map dag}=[draw,shape=SEbox,inner sep=2pt,minimum height=6mm,fill=white,minimum width=7mm]
\tikzstyle{medium map adj}=[draw,shape=SEbox,inner sep=2pt,minimum height=6mm,fill=white,minimum width=7mm]
\tikzstyle{medium map trans}=[draw,shape=SWbox,inner sep=2pt,minimum height=6mm,fill=white,minimum width=7mm]
\tikzstyle{medium map conj}=[draw,shape=NWbox,inner sep=2pt,minimum height=6mm,fill=white,minimum width=7mm]
\tikzstyle{semilarge map}=[draw,shape=NEbox,inner sep=2pt,minimum height=6mm,fill=white,minimum width=9.5mm]
\tikzstyle{semilarge map trans}=[draw,shape=SWbox,inner sep=2pt,minimum height=6mm,fill=white,minimum width=9.5mm]
\tikzstyle{semilarge map adj}=[draw,shape=SEbox,inner sep=2pt,minimum height=6mm,fill=white,minimum width=9.5mm]
\tikzstyle{semilarge map dag}=[draw,shape=SEbox,inner sep=2pt,minimum height=6mm,fill=white,minimum width=9.5mm]
\tikzstyle{semilarge map conj}=[draw,shape=NWbox,inner sep=2pt,minimum height=6mm,fill=white,minimum width=9.5mm]
\tikzstyle{large map}=[draw,shape=NEbox,inner sep=2pt,minimum height=6mm,fill=white,minimum width=12mm]
\tikzstyle{large map conj}=[draw,shape=NWbox,inner sep=2pt,minimum height=6mm,fill=white,minimum width=12mm]
\tikzstyle{very large map}=[draw,shape=NEbox,inner sep=2pt,minimum height=6mm,fill=white,minimum width=17mm]

\tikzstyle{medium dmap}=[draw,doubled,shape=NEbox,inner sep=2pt,minimum height=6mm,fill=white,minimum width=7mm]
\tikzstyle{medium dmap dag}=[draw,doubled,shape=SEbox,inner sep=2pt,minimum height=6mm,fill=white,minimum width=7mm]
\tikzstyle{medium dmap adj}=[draw,doubled,shape=SEbox,inner sep=2pt,minimum height=6mm,fill=white,minimum width=7mm]
\tikzstyle{medium dmap trans}=[draw,doubled,shape=SWbox,inner sep=2pt,minimum height=6mm,fill=white,minimum width=7mm]
\tikzstyle{medium dmap conj}=[draw,doubled,shape=NWbox,inner sep=2pt,minimum height=6mm,fill=white,minimum width=7mm]
\tikzstyle{semilarge dmap}=[draw,doubled,shape=NEbox,inner sep=2pt,minimum height=6mm,fill=white,minimum width=9.5mm]
\tikzstyle{semilarge dmap trans}=[draw,doubled,shape=SWbox,inner sep=2pt,minimum height=6mm,fill=white,minimum width=9.5mm]
\tikzstyle{semilarge dmap adj}=[draw,doubled,shape=SEbox,inner sep=2pt,minimum height=6mm,fill=white,minimum width=9.5mm]
\tikzstyle{semilarge dmap dag}=[draw,doubled,shape=SEbox,inner sep=2pt,minimum height=6mm,fill=white,minimum width=9.5mm]
\tikzstyle{semilarge dmap conj}=[draw,doubled,shape=NWbox,inner sep=2pt,minimum height=6mm,fill=white,minimum width=9.5mm]
\tikzstyle{large dmap}=[draw,doubled,shape=NEbox,inner sep=2pt,minimum height=6mm,fill=white,minimum width=12mm]
\tikzstyle{large dmap conj}=[draw,doubled,shape=NWbox,inner sep=2pt,minimum height=6mm,fill=white,minimum width=12mm]
\tikzstyle{large dmap trans}=[draw,doubled,shape=SWbox,inner sep=2pt,minimum height=6mm,fill=white,minimum width=12mm]
\tikzstyle{large dmap adj}=[draw,doubled,shape=SEbox,inner sep=2pt,minimum height=6mm,fill=white,minimum width=12mm]
\tikzstyle{large dmap dag}=[draw,doubled,shape=SEbox,inner sep=2pt,minimum height=6mm,fill=white,minimum width=12mm]
\tikzstyle{very large dmap}=[draw,doubled,shape=NEbox,inner sep=2pt,minimum height=6mm,fill=white,minimum width=19.5mm]

\tikzstyle{muxbox}=[draw,shape=rectangle,minimum height=3mm,minimum width=3mm,fill=white]
\tikzstyle{dmuxbox}=[muxbox,doubled]

\tikzstyle{box}=[draw,shape=rectangle,inner sep=2pt,minimum height=6mm,minimum width=6mm,fill=white]
\tikzstyle{dbox}=[draw,doubled,shape=rectangle,inner sep=2pt,minimum height=6mm,minimum width=6mm,fill=white]
\tikzstyle{dmap}=[draw,doubled,shape=NEbox,inner sep=2pt,minimum height=6mm,fill=white]
\tikzstyle{dmapdag}=[draw,doubled,shape=SEbox,inner sep=2pt,minimum height=6mm,fill=white]
\tikzstyle{dmapadj}=[draw,doubled,shape=SEbox,inner sep=2pt,minimum height=6mm,fill=white]
\tikzstyle{dmaptrans}=[draw,doubled,shape=SWbox,inner sep=2pt,minimum height=6mm,fill=white]
\tikzstyle{dmapconj}=[draw,doubled,shape=NWbox,inner sep=2pt,minimum height=6mm,fill=white]

\tikzstyle{ddmap}=[draw,doubled,dashed,shape=NEbox,inner sep=2pt,minimum height=6mm,fill=white]
\tikzstyle{ddmapdag}=[draw,doubled,dashed,shape=SEbox,inner sep=2pt,minimum height=6mm,fill=white]
\tikzstyle{ddmapadj}=[draw,doubled,dashed,shape=SEbox,inner sep=2pt,minimum height=6mm,fill=white]
\tikzstyle{ddmaptrans}=[draw,doubled,dashed,shape=SWbox,inner sep=2pt,minimum height=6mm,fill=white]
\tikzstyle{ddmapconj}=[draw,doubled,dashed,shape=NWbox,inner sep=2pt,minimum height=6mm,fill=white]

\boxshape{sNEbox}{0pt}{3pt}
\boxshape{sSEbox}{0pt}{-3pt}
\boxshape{sNWbox}{3pt}{0pt}
\boxshape{sSWbox}{-3pt}{0pt}
\tikzstyle{smap}=[draw,shape=sNEbox,fill=white]
\tikzstyle{smapdag}=[draw,shape=sSEbox,fill=white]
\tikzstyle{smapadj}=[draw,shape=sSEbox,fill=white]
\tikzstyle{smaptrans}=[draw,shape=sSWbox,fill=white]
\tikzstyle{smapconj}=[draw,shape=sNWbox,fill=white]

\tikzstyle{dsmap}=[draw,dashed,shape=sNEbox,fill=white]
\tikzstyle{dsmapdag}=[draw,dashed,shape=sSEbox,fill=white]
\tikzstyle{dsmaptrans}=[draw,dashed,shape=sSWbox,fill=white]
\tikzstyle{dsmapconj}=[draw,dashed,shape=sNWbox,fill=white]

\boxshape{mNEbox}{0pt}{10pt}
\boxshape{mSEbox}{0pt}{-10pt}
\boxshape{mNWbox}{10pt}{0pt}
\boxshape{mSWbox}{-10pt}{0pt}
\tikzstyle{mmap}=[draw,shape=mNEbox]
\tikzstyle{mmapdag}=[draw,shape=mSEbox]
\tikzstyle{mmaptrans}=[draw,shape=mSWbox]
\tikzstyle{mmapconj}=[draw,shape=mNWbox]

\tikzstyle{mmapgray}=[draw,fill=gray!40!white,shape=mNEbox]
\tikzstyle{smapgray}=[draw,fill=gray!40!white,shape=sNEbox]

\makeatletter

\pgfdeclareshape{cornerpoint}{
\inheritsavedanchors[from=rectangle] 
\inheritanchorborder[from=rectangle]
\inheritanchor[from=rectangle]{center}
\inheritanchor[from=rectangle]{north}
\inheritanchor[from=rectangle]{south}
\inheritanchor[from=rectangle]{west}
\inheritanchor[from=rectangle]{east}
\backgroundpath{
\southwest \pgf@xa=\pgf@x \pgf@ya=\pgf@y
\northeast \pgf@xb=\pgf@x \pgf@yb=\pgf@y

\pgfmathsetmacro{\pgf@shorten@left}{\pgfkeysvalueof{/tikz/shorten left}}
\pgfmathsetmacro{\pgf@shorten@right}{\pgfkeysvalueof{/tikz/shorten right}}

\pgfpathmoveto{\pgfpoint{0.5 * (\pgf@xa + \pgf@xb)}{\pgf@ya - 5pt}}
\pgfpathlineto{\pgfpoint{\pgf@xa - 8pt + \pgf@shorten@left}{\pgf@yb - 1.5 * \pgf@shorten@left}}
\pgfpathlineto{\pgfpoint{\pgf@xa - 8pt + \pgf@shorten@left}{\pgf@yb}}
\pgfpathlineto{\pgfpoint{\pgf@xb + 8pt - \pgf@shorten@right}{\pgf@yb}}
\pgfpathlineto{\pgfpoint{\pgf@xb + 8pt - \pgf@shorten@right}{\pgf@yb - 1.5 * \pgf@shorten@right}}
\pgfpathclose
}
}

\pgfdeclareshape{cornercopoint}{
\inheritsavedanchors[from=rectangle] 
\inheritanchorborder[from=rectangle]
\inheritanchor[from=rectangle]{center}
\inheritanchor[from=rectangle]{north}
\inheritanchor[from=rectangle]{south}
\inheritanchor[from=rectangle]{west}
\inheritanchor[from=rectangle]{east}
\backgroundpath{
\southwest \pgf@xa=\pgf@x \pgf@ya=\pgf@y
\northeast \pgf@xb=\pgf@x \pgf@yb=\pgf@y

\pgfmathsetmacro{\pgf@shorten@left}{\pgfkeysvalueof{/tikz/shorten left}}
\pgfmathsetmacro{\pgf@shorten@right}{\pgfkeysvalueof{/tikz/shorten right}}

\pgfpathmoveto{\pgfpoint{0.5 * (\pgf@xa + \pgf@xb)}{\pgf@yb + 5pt}}
\pgfpathlineto{\pgfpoint{\pgf@xa - 8pt + \pgf@shorten@left}{\pgf@ya + 1.5 * \pgf@shorten@left}}
\pgfpathlineto{\pgfpoint{\pgf@xa - 8pt + \pgf@shorten@left}{\pgf@ya}}
\pgfpathlineto{\pgfpoint{\pgf@xb + 8pt - \pgf@shorten@right}{\pgf@ya}}
\pgfpathlineto{\pgfpoint{\pgf@xb + 8pt - \pgf@shorten@right}{\pgf@ya + 1.5 * \pgf@shorten@right}}
\pgfpathclose
}
}

\makeatother

\pgfkeyssetvalue{/tikz/shorten left}{0pt}
\pgfkeyssetvalue{/tikz/shorten right}{0pt}

\tikzstyle{kpoint common}=[draw,fill=white,inner sep=1pt,minimum height=4mm]
\tikzstyle{kpoint sc}=[shape=cornerpoint,kpoint common]
\tikzstyle{kpoint adjoint sc}=[shape=cornercopoint,kpoint common]
\tikzstyle{kpoint}=[shape=cornerpoint,shorten left=5pt,kpoint common]
\tikzstyle{kpoint adjoint}=[shape=cornercopoint,shorten left=5pt,kpoint common]
\tikzstyle{kpoint conjugate}=[shape=cornerpoint,shorten right=5pt,kpoint common]
\tikzstyle{kpoint transpose}=[shape=cornercopoint,shorten right=5pt,kpoint common]
\tikzstyle{kpoint symm}=[shape=cornerpoint,shorten left=5pt,shorten right=5pt,kpoint common]

\tikzstyle{wide kpoint sc}=[shape=cornerpoint,kpoint common, minimum width=1 cm]
\tikzstyle{wide kpointdag sc}=[shape=cornercopoint,kpoint common, minimum width=1 cm]

\tikzstyle{black kpoint}=[shape=cornerpoint,shorten left=5pt,kpoint common,fill=black,font=\color{white}]

\tikzstyle{black kpoint sm}=[shape=cornerpoint,shorten left=5pt,kpoint common,fill=black,font=\color{white},scale=0.75]

\tikzstyle{black kpoint adjoint}=[shape=cornercopoint,shorten left=5pt,kpoint common,fill=black,font=\color{white}]
\tikzstyle{black kpointadj}=[shape=cornercopoint,shorten left=5pt,kpoint common,fill=black,font=\color{white}]

\tikzstyle{black kpointadj sm}=[shape=cornercopoint,shorten left=5pt,kpoint common,fill=black,font=\color{white},scale=0.75]

\tikzstyle{black dkpoint}=[shape=cornerpoint,shorten left=5pt,kpoint common,fill=black, doubled,font=\color{white}]
\tikzstyle{black dkpoint adjoint}=[shape=cornercopoint,shorten left=5pt,kpoint common,fill=black, doubled,font=\color{white}]
\tikzstyle{black dkpointadj}=[shape=cornercopoint,shorten left=5pt,kpoint common,fill=black, doubled,font=\color{white}]

\tikzstyle{black dkpoint sm}=[shape=cornerpoint,shorten left=5pt,kpoint common,fill=black, doubled,font=\color{white},scale=0.75]
\tikzstyle{black dkpointadj sm}=[shape=cornercopoint,shorten left=5pt,kpoint common,fill=black, doubled,font=\color{white},scale=0.75] 

\tikzstyle{kpointdag}=[kpoint adjoint]
\tikzstyle{kpointadj}=[kpoint adjoint]
\tikzstyle{kpointconj}=[kpoint conjugate]
\tikzstyle{kpointtrans}=[kpoint transpose]

\tikzstyle{big kpoint}=[kpoint, minimum width=1.2 cm, minimum height=8mm, inner sep=4pt, text depth=3mm]

\tikzstyle{wide kpoint}=[kpoint, minimum width=1 cm, inner sep=2pt]
\tikzstyle{wide kpointdag}=[kpointdag, minimum width=1 cm, inner sep=2pt]
\tikzstyle{wide kpointconj}=[kpointconj, minimum width=1 cm, inner sep=2pt]
\tikzstyle{wide kpointtrans}=[kpointtrans, minimum width=1 cm, inner sep=2pt]

\tikzstyle{wider kpoint}=[kpoint, minimum width=1.25 cm, inner sep=2pt]
\tikzstyle{wider kpointdag}=[kpointdag, minimum width=1.25 cm, inner sep=2pt]
\tikzstyle{wider kpointconj}=[kpointconj, minimum width=1.25 cm, inner sep=2pt]
\tikzstyle{wider kpointtrans}=[kpointtrans, minimum width=1.25 cm, inner sep=2pt]

\tikzstyle{gray kpoint}=[kpoint,fill=gray!50!white]
\tikzstyle{gray kpointdag}=[kpointdag,fill=gray!50!white]
\tikzstyle{gray kpointadj}=[kpointadj,fill=gray!50!white]
\tikzstyle{gray kpointconj}=[kpointconj,fill=gray!50!white]
\tikzstyle{gray kpointtrans}=[kpointtrans,fill=gray!50!white]

\tikzstyle{gray dkpoint}=[kpoint,fill=gray!50!white,doubled]
\tikzstyle{gray dkpointdag}=[kpointdag,fill=gray!50!white,doubled]
\tikzstyle{gray dkpointadj}=[kpointadj,fill=gray!50!white,doubled]
\tikzstyle{gray dkpointconj}=[kpointconj,fill=gray!50!white,doubled]
\tikzstyle{gray dkpointtrans}=[kpointtrans,fill=gray!50!white,doubled]

\tikzstyle{white label}=[draw,fill=white,rectangle,inner sep=0.7 mm]
\tikzstyle{gray label}=[draw,fill=gray!50!white,rectangle,inner sep=0.7 mm]
\tikzstyle{black label}=[draw,fill=black,rectangle,inner sep=0.7 mm]

\tikzstyle{dkpoint}=[kpoint,doubled]
\tikzstyle{wide dkpoint}=[wide kpoint,doubled]
\tikzstyle{dkpointdag}=[kpoint adjoint,doubled]
\tikzstyle{wide dkpointdag}=[wide kpointdag,doubled]
\tikzstyle{dkcopoint}=[kpoint adjoint,doubled]
\tikzstyle{dkpointadj}=[kpoint adjoint,doubled]
\tikzstyle{dkpointconj}=[kpoint conjugate,doubled]
\tikzstyle{dkpointtrans}=[kpoint transpose,doubled]

\tikzstyle{kscalar}=[kpoint common, shape=EBox, inner xsep=-1pt, inner ysep=3pt,font=\small]
\tikzstyle{kscalarconj}=[kpoint common, shape=WBox, inner xsep=-1pt, inner ysep=3pt,font=\small]

\tikzstyle{spekpoint}=[kpoint sc,minimum height=5mm,inner sep=3pt]
\tikzstyle{spekcopoint}=[kpoint adjoint sc,minimum height=5mm,inner sep=3pt]

\tikzstyle{dspekpoint}=[spekpoint,doubled]
\tikzstyle{dspekcopoint}=[spekcopoint,doubled]

\tikzset{->-/.style={decoration={
  markings,
  mark=at position .5 with {\arrow{>}}},postaction={decorate}}}

 \tikzstyle{upground}=[circuit ee IEC,thick,ground,rotate=90,scale=1.5]
 \tikzstyle{downground}=[circuit ee IEC,thick,ground,rotate=-90,scale=1.5]
 \tikzstyle{bigground}=[regular polygon,regular polygon sides=3,draw=gray,scale=0.50,inner sep=-0.5pt,minimum width=10mm,fill=gray]

\tikzstyle{arrs}=[-latex,font=\small,auto]
\tikzstyle{arrow plain}=[arrs]
\tikzstyle{arrow dashed}=[dashed,arrs]
\tikzstyle{arrow bold}=[very thick,arrs]
\tikzstyle{arrow hide}=[draw=white!0,-]
\tikzstyle{arrow reverse}=[latex-]
\tikzstyle{cdnode}=[]
\tikzstyle{diredges}=[every to/.style={diredge}]

\tikzstyle{discarding}=[fill=white, draw=black, shape=circle, style=upground]
\tikzstyle{state}=[fill=white, draw=black, style=point, tikzit shape=rectangle]
\tikzstyle{effect}=[fill=white, draw=black, shape=circle]
\tikzstyle{dstate}=[fill=white, draw=black, style=dpoint, tikzit shape=rectangle]
\tikzstyle{deffect}=[fill=white, draw=black, shape=circle, style=dcopoint]
\tikzstyle{dottededge}=[-, dotted]
\tikzstyle{dottededge }=[-, draw=red, tikzit draw=red]
\tikzstyle{double edge}=[-, style=boldedge, draw=black, tikzit draw={rgb,255: red,191; green,0; blue,64}]
\tikzstyle{new edge style 0}=[-]
\tikzstyle{diredge}=[-, style=diredge]
\tikzstyle{new edge style 1}=[-, style=diredge, dashed]

\tikzstyle{discarding}=[fill=white, draw=black, shape=circle, style=upground]
\tikzstyle{state}=[fill=white, draw=black, style=point, tikzit shape=rectangle]
\tikzstyle{effect}=[fill=white, draw=black, shape=circle]
\tikzstyle{dstate}=[fill=white, draw=black, style=dpoint, tikzit shape=rectangle]
\tikzstyle{deffect}=[fill=white, draw=black, shape=circle, style=dcopoint]

\tikzstyle{dottededge}=[-, dotted]
\tikzstyle{dottededge }=[-, draw=red, tikzit draw=red]
\tikzstyle{double edge}=[-, style=boldedge, draw=black, tikzit draw={rgb,255: red,191; green,0; blue,64}]
\tikzstyle{new edge style 0}=[-]
\tikzstyle{diredge}=[-, style=diredge]
\tikzstyle{new edge style 1}=[-, style=diredge, dashed]
\tikzstyle{doubledge}=[-, tikzit draw={rgb,255: red,15; green,0; blue,228}]
\tikzstyle{arrow}=[-, style=->-]
\tikzstyle{arrow dashed}=[-, style=->-, dashed]
\tikzstyle{dashed2}=[-, dashed]
\tikzstyle{dashedbold}=[-, style=boldedge, dashed]

\def\bR{\begin{color}{red}}
\def\bB{\begin{color}{blue}}
\def\bM{\begin{color}{magenta}} 
\def\bC{\begin{color}{cyan}}
\def\bW{\begin{color}{white}}
\def\bBl{\begin{color}{black}}
\def\bG{\begin{color}{green}}
\def\bY{\begin{color}{yellow}}
\def\e{\end{color}\xspace}

\newcommand{\tikzfigscale}[2]{\scalebox{#1}{\input{#2.tikz}}}

\begin{document}

\title{Categorical Semantics for Time Travel}

\author{
	\IEEEauthorblockN{Nicola Pinzani}
	\IEEEauthorblockA{Quantum Group\\
	University of Oxford\\
	\texttt{nicola.pinzani@cs.ox.ac.uk}
	}
	\and
	\IEEEauthorblockN{Stefano Gogioso}
	\IEEEauthorblockA{Quantum Group\\
	University of Oxford\\
	\texttt{stefano.gogioso@cs.ox.ac.uk}
	}
	\and
	\IEEEauthorblockN{Bob Coecke}
	\IEEEauthorblockA{Quantum Group\\
	University of Oxford\\
	\texttt{bob.coecke@cs.ox.ac.uk}
	}
}

\maketitle

\begin{abstract}
	We introduce a general categorical framework to reason about quantum theory and other process theories living in spacetimes where Closed Timelike Curves (CTCs) are available, allowing resources to travel back in time and provide computational speedups. Our framework is based on a weakening of the definition of traced symmetric monoidal categories, obtained by dropping the yanking axiom and the requirement that the trace be defined on all morphisms.
	We show that the two leading models for quantum theory with closed timelike curves---namely the P-CTC model of Lloyd et al. and the D-CTC model of Deutsch---are captured by our framework, and in doing so we provide the first compositional description of the D-CTC model. Our description of the D-CTC model results in a process theory which respects the constraints of relativistic causality: this is in direct contrast to the P-CTC model, where CTCs are implemented by a trace and allow post-selection to be performed deterministically.
\end{abstract}

\section{Introduction}
\label{section:introduction}

\noindent The possibility of traveling back in time, and the many paradoxes associated with its more practical formulations, have fascinated humans for centuries, and the development of Relativity provided a solid mathematical foundation to the concept in the form of Closed Timelike Curves (CTCs) \cite{godel1949,kerr1963,misner1965,bonnor1980}. In the context of quantum computer science, the possibility of time travel acquires a special significance, because models of quantum processes in the presence of CTCs display large computational speedups.

For a detailed discussion of various quantum time-travel models, we refer the interested reader to \cite{allen2014} and references therein. Here, we focus our attention to two specific models: the \emph{D-CTC model} of Deutsch \cite{deutsch1991} and the \emph{P-CTC model} of Lloyd et al. \cite{lloyd2011a,lloyd2011b}.
In the D-CTC model, interaction with quantum systems living on a time loop is captured by a fixed-point operation, while in the P-CTC model the same phenomenon is captured by projection. It was shown \cite{aaronson2005,aaronson2009} that the presence of CTCs brings significant speedups in both cases: in the D-CTC model, quantum computers having access to CTCs can solve all problems of the PSPACE class\footnote{The class of problems which can be decided in polynomial space by deterministic Turing machines.} in polynomial time, while in the P-CTC model the same computers can solve all problems of the PP class\footnote{The class of problems which can be decided in polynomial time by probabilistic Turing machines with an error probability less than 1/2.} in polynomial time (essentially because it can be shown that P-CTCs are equivalent to post-selection \cite{brun2012}). Though no proof of separation between BQP\footnote{The class of problems which can be decided in polynomial time by quantum computers with a single-run error probability of at most 1/3.}, PP and PSPACE exists, the inclusions between the three classes are believed to be strict, so that the use of CTCs in quantum computation would bring significant advantage.

Our approach differs from that of other works on time travel in that we won't debate the possibility of \emph{actors} traveling back in time, but rather we will limit ourselves to \emph{computational resources} traveling back in time. Our actors live in a \emph{chronology-respecting} (CR) region---i.e. in a region of spacetime where the usual Relativistic laws of causality and no-signaling apply---and they interact with time-traveling resources in a local fashion, across a \emph{Cauchy horizon} which separates the CR region from the \emph{chronology-violating} (CV) region, which contains the CTCs. There are several reasons behind this particular take on time travel, but one factor playing a heavy role in our decision was the existence of a large body of research suggesting that Cauchy horizons would not realistically be crossable by anything as heavy as a practicing computer scientist \cite{simpson1973,chandrasekhar1982,morris1988,hawking1992,rauch1993,hintz2017}.

This work fits within the broader framework of \emph{process theories} (aka symmetric monoidal categories) and \emph{categorical quantum mechanics} \cite{abramsky2009,selinger2011,coecke_kissinger_book}; see Appendix \ref{appendix_graphical} for more details about the graphical notation. More specifically, it is part of recent efforts to understand the complex interplay between quantum theory and Relativistic causal structure, initiated by \cite{coecke2013,coecke2016} and recently brought into the spotlight by the work of \cite{kissinger2017} on higher causal structure. Here, we push the envelope and give a rigorous process-theoretic treatment of chronology-violating causal scenarios: as previously mentioned, these are of great interest for their complexity implications on quantum computing, but they are also easy to misunderstand and riddled with paradoxes. The development of sound categorical semantics brings general reasoning about such scenarios back on firm ground.

In Section \ref{section:lloyd}, we briefly introduce a process-theoretic treatment of the P-CTC model based on previous work by Oreshkov and Cerf \cite{oreshkov2015,oreshkov2016} on time-symmetric quantum theory.
In Section \ref{section:deutsch}, we introduce the first process-theoretic treatment of the D-CTC model to date.
In Section \ref{section:semantics}, we introduce categorical semantics for process theories involving time travel, defining certain super-operators---with properties akin to those defining traced monoidal categories---which completely encapsulate the effects of interaction between the chronology-respecting region and the CTCs.

The super-operators defined in Section \ref{section:semantics} are different from traces in two key aspects. Firstly, they do are not required to satisfy the \textit{yanking axiom}, which we show to be violated by the D-CTC model. Secondly, we don't require traces to be defined on all morphisms, but rather only a monoidal sub-category of them: this is a consequence of our requirement that the interaction with CTCs be fully localized in spacetime.

\section{The P-CTC model}
\label{section:lloyd}

\subsection{Single-system Definition}

\noindent In their P-CTC model \cite{lloyd2011a,lloyd2011b}, Lloyd et al.~propose to use post-selection to simulate CTCs, a development inspired by the graphical treatments of quantum teleportation \cite{LE1,svetlichny2011}.
The P-CTC construction was originally designed with unitary interactions in mind, the proposal being to describe the evolution of a chronology-respecting system in a state $\rho$ using the following transformation:
\begin{equation}
	\rho' := \dfrac{E \rho E^\dagger}{Tr[E \rho E^\dagger]}
\end{equation}
where the operator $E$ is defined to be $E \colon = Tr_{CV}[U]$ and $U$ is the unitary interaction between the chronology-respecting (CR) and chronology-violating (CV) regions.
The definition can be straightforwardly extended to CPTP maps, and the evolution is represented in the graphical calculus as follows:
\begin{equation}
	\tikzfigscale{0.7}{staterho}
	\hspace{2mm}
	\mapsto
	\hspace{2mm}
	\frac{1}{Tr\big[Tr_{CV}[f] \circ \rho\big]}
	\tikzfigscale{0.7}{loyd_evolution}
\end{equation}

The transformation $\rho \mapsto Tr_{CV}[f] \circ \rho$ is linear, but it is not generally trace-preserving, so we renormalize the interaction to make sure that the resulting state $\rho'$ is a density matrix. In those cases where renormalization cannot be performed---i.e. when $\rho'$ is the zero state---the model assumes that the interaction simply cannot happen: in this sense, the P-CTC model of CTCs can be thought to post-select those outcomes which do not lead to contradictory time-travel stories. In fact, the P-CTC model is equivalent to post-selection \cite{brun2012}, an observation which we can use to quickly derive a categorical model.

\subsection{A category for the P-CTC model}

\noindent In recent work, Oreshkov and Cerf \cite{oreshkov2015,oreshkov2016} developed a process-theoretic treatment of time-symmetric quantum theory, which we use as a categorical model of quantum theory in the presence of the P-CTCs of the Lloyd model.

\begin{definition}
	The symmetric monoidal category $\Mix_{sym}$ has finite-dimensional Hilbert spaces as its objects. The morphisms $f: \mathcal{A} \rightarrow \mathcal{B}$ in $\Mix_{sym}$ are chosen to be the zero CP map together with all the CP maps satisfying the following condition:
	\begin{equation}
		\label{eqn:condition_mixsym}
		Tr \left[f\left(\dfrac{\mathbb{1}}{d_{\mathcal{A}}}\right)\right] \;=\; 1
	\end{equation}
	that is, written diagrammatically:
	\begin{equation}
		\dfrac{1}{d_{\mathcal{A}}} \ \ \tikzfigscale{0.8}{time_symmetric}
		\hspace{2mm}
		=
		\hspace{2mm}
		\tikzfigscale{0.8}{empty_diagram}
	\end{equation}
	The tensor product in $\Mix_{sym}$ is that of CP maps. Sequential composition in $\Mix_{sym}$ is defined as:
	\begin{equation}
		\label{eqn:rule_composition}
		f \circ_{\Mix_{sym}} g \;:=\; \dfrac{f \circ g}{Tr \left[(f \circ g)\left(\dfrac{\mathbb{1}}{d_{\mathcal{A}}}\right)\right] }
	\end{equation}
	if the normalizability condition $Tr \left[(f \circ g)\left(\dfrac{\mathbb{1}}{d_{\mathcal{A}}}\right)\right] \neq 0$ holds, and defined to be $f \circ_{\Mix_{sym}} g := 0$ otherwise.
\end{definition}

Essentially, $\Mix_{sym}$ provides a categorical model for post-selected quantum theory. Within it, the P-CTC prescription for interaction with CTCs is readily realized as follows:
\begin{equation}
	\tikzfigscale{0.8}{staterho}
	\hspace{3mm}
	\mapsto
	\hspace{3mm}
	\tikzfigscale{0.8}{composition_PCTC}
\end{equation}
The equivalence of P-CTCs with post-selection means that the P-CTC model necessarily violates the Relativistic constraints of causality and no-signaling;  in the parlance of \cite{coecke2013,coecke2016}, the category $\Mix_{sym}$ is not terminal. This violation of the laws of causality may seem natural---even necessary---in the presence of time-travel: in the next Section, we shall see that this is not actually the case.

\section{The D-CTC model}
\label{section:deutsch}

\subsection{Single-system Definition}

\noindent Deutsch's D-CTC model \cite{deutsch1991} describes local interactions between a quantum system in a CR region---which can take part in other operations both in the past and in the future---and a quantum system in a CV region---which is only available as part of the single local interaction at hand. The model as originally formulated is single-system, i.e. it does not directly deal with the issue of composing various interactions in parallel. The generic single-system process in the D-CTC model can be written in the following form, where $\Phi: \mathcal{H} \otimes \mathcal{C} \rightarrow \mathcal{K} \otimes \mathcal{C}$ is a CPTP map: 
\begin{equation}
	\tikzfigscale{0.8}{general}
\end{equation}
The markings on the right denote that the state on quantum system $\mathcal{C}$ emerges from the CTC in the immediate past of the interaction, and re-enters the CTC in its immediate future, never to be accessed again.

Deutsch defines the single-system behavior of an interaction with CTCs as a function\footnote{Not necessarily linear, nor necessarily continuous.} mapping normalized quantum states on $\mathcal{H}$ to normalized quantum states on $\mathcal{K}$. Specifically, given a normalized input state $\rho$ on $\mathcal{H}$, the corresponding output state is defined as follows:
\begin{equation}
	\tikzfigscale{0.8}{generalfixedpoint2}
	\hspace{2mm}
	\mapsto
	\hspace{3mm}
	\tikzfigscale{0.8}{generalfixedpoint3}
\end{equation}
where $\tau$ is the normalized quantum state on $\mathcal{C}$ of maximal entropy amongst those satisfying the following fixed-point equation:
\begin{equation}
	\tikzfigscale{0.8}{generalfixedpoint}
	\hspace{3mm}
	=
	\hspace{3mm}
	\tikzfigscale{0.8}{generalfixedpoint1}
\end{equation}
The existence of such fixed points is a consequence of Brouwer's fixed-point theorem on convex compact subsets of Euclidean space, and the fact that CPTP maps send the convex compact set of normalized states of Hilbert space $\mathcal{C}$ into itself. A geometric characterization of the fixed point of maximal entropy is given in Appendix \ref{appendix:fixedPoint}.

\subsection{The grandfather's paradox}

\noindent In order to show the D-CTC construction at work, we look at the most iconic of all time-travel paradoxes: the \emph{grandfather's paradox}. In the {grandfather's paradox}, a time-traveler goes back in time and kills their own grandfather, thus preventing themselves from being born in the first place: but if they were never born, they never could have gone back in time to kill their own grandfather, and therefore they would have actually been born instead. A basic implementation of the grandfather paradox with qubits and CTCs involves a CNOT gate followed by a swap:
\begin{equation}
\label{equation:timetravelCNOT}
	\tikzfigscale{0.8}{timetravelCNOT}
\end{equation}
The Z basis is be used to encode and measure the state of the ``individuals'' involved in the paradox: the ``time-traveler's grandfather'' is represented by the input qubit in the CR region, the ``time-traveler'' is represented by the output qubit in the CR region. The state $\ket{1}$ is taken to mean ``alive'' and the state $\ket{0}$ is taken to mean ``dead''. the D-CTC model tasks us with solving the following fixed point equation, where the gray $\pi$ dot is the Pauli X gate, i.e.~the unary NOT gate for the Z basis:
\begin{equation}
	\tikzfigscale{0.8}{deutschsolution}
\end{equation}
The only solution to the above fixed-point equation is given by taking $\tau$ to the the maximally mixed state, which we can plug back into the original map:
\begin{equation}
	\tikzfigscale{0.8}{deutschsolutionCR1}
\end{equation}
This yields the following resolution to the grandfather's paradox in the D-CTC model, where the time-traveler is in a totally mixed state of alive and dead at the end of the affair, with 50\% probability of each being true:
\begin{equation}
	\tikzfigscale{0.8}{timetravelCNOT}
	\hspace{5mm}
	=
	\hspace{5mm}
	\scalebox{0.8}{
	\begin{tikzpicture}
		\begin{pgfonlayer}{nodelayer}
			\node [style=downground] (0) at (0, -1) {};
			\node [style=none] (1) at (0, +1) {};
			\node [style=none] (2) at (0.5, -1) {$\frac{1}{2}$};
		\end{pgfonlayer}
		\begin{pgfonlayer}{edgelayer}
			\draw [style=boldedge] (0) to (1);
		\end{pgfonlayer}
	\end{tikzpicture}
	}
\end{equation}

\subsection{Entanglement Breaking}

\noindent The example of the grandfather's paradox involves a single system in the CR region, but what about composite systems? What if we have two qubits in an entangled state living in the CR region? Such an example leads us to discover the first weird property of the D-CTC model: interaction with D-CTCs breaks entanglement. Indeed, consider the following scenario involving two CR systems and one CV system:
\begin{equation}
\label{equation:breakingentanglement}
	\dfrac{1}{2} \ \ \tikzfigscale{0.8}{bell}
	\hspace{3mm}
	\mapsto
	\hspace{3mm}
	\tikzfigscale{0.6}{breakingentanglement}
\end{equation}
The fixed-point equation prescribed by the D-CTC model to evaluate the state above can be expanded as follows:
\begin{equation}
\scalebox{0.7}{$
	\tikzfigscale{1}{tau}
	\hspace{1mm}
	=
	\hspace{1mm}
	\tikzfigscale{1}{breakingentanglementfixedpoint0}
	\hspace{1mm}
	=
	\hspace{1mm}
	\tikzfigscale{1}{breakingentanglementfixedpoint1}
	\hspace{1mm}
	=
	\hspace{1mm}
	\tikzfigscale{1}{breakingentanglementfixedpoint2}
	\hspace{1mm}
	=
	\hspace{1mm}
	\tikzfigscale{1}{breakingentanglementfixedpoint3}
$}
\end{equation}
As a consequence, the state in \ref{equation:breakingentanglement} breaks into the totally mixed state on the two CR systems:
\begin{equation}
	\dfrac{1}{2} \ \ \tikzfigscale{0.8}{bell}
	\hspace{3mm}
	\mapsto
	\hspace{3mm}
	\frac{1}{2}\tikzfigscale{0.8}{noise} \ \ \frac{1}{2}\tikzfigscale{0.8}{noise}
\end{equation}
More generally, it is possible to show that the same map breaks entanglement of all bipartite states:
\begin{equation}
	\tikzfigscale{0.6}{breakingentanglementgeneral1}
	\hspace{1mm}
	=
	\hspace{1mm}
	\tikzfigscale{0.6}{breakingentanglementgeneral2}
	\hspace{-1mm}
	=
	\hspace{1mm}
	\tikzfigscale{0.6}{breakingentanglementgeneral3}
\end{equation}
Coupled with the use of informationally complete measurements, this effect can be used to produce \emph{approximate} clones of an arbitrary quantum state, approaching perfect fidelity as the dimension of the CR system used is increased; for more details on this approach to cloning, see Appendix \ref{appendix:cloning}.

Here, we are more interested in a another, subtler consequence of entanglement breaking, namely the fact that the D-CTC model is not \emph{locally process tomographic}: it is not possible to identify a map by looking at the outputs it gives on single-system inputs. More specifically, there are maps which act identically when they are applied to single-system input states:
\begin{equation} 
	\tikzfigscale{0.8}{localtomography1}
	\hspace{3mm}
	=
	\hspace{3mm}
	\tikzfigscale{0.8}{localtomography2}
	\hspace{5mm}
	\text{for all states $\rho$}
\end{equation}
but act entirely differently when applied to a sub-system of some larger entangled state:
\begin{equation}
\label{equation:failureOfTomography}
	\tikzfigscale{0.6}{localtomography3}
	\hspace{3mm}
	\neq
	\hspace{3mm}
	\tikzfigscale{0.6}{localtomography4}
\end{equation}
This means that, unlike CPTP maps, D-CTC maps are not functions on single-system states, something which we will need to take into account when defining the category \DMix in the next subsection.

\subsection{A category for the D-CTC model}

\noindent In order to define a category \DMix for the D-CTC model, we first need to set some terminology and notation. By an \emph{elementary morphism} $\mathcal{H} \rightarrow \mathcal{K}$ in \DMix we will mean a process in the following form, where $\Phi: \mathcal{H} \otimes \mathcal{C} \rightarrow \mathcal{K} \otimes \mathcal{C}$ is a CPTP map:
\begin{equation}
	\tikzfigscale{0.8}{general}
\end{equation}
There are two pieces of data in the elementary morphism above, the CPTP map $\Phi$ capturing the interaction and the system $\mathcal{C}$ living in the CV region. Because local process tomography fails in the D-CTC model, we cannot rely on functions as a substrate to define our category on: instead, we will define it in terms of sequences of elementary morphisms, quotiented by an appropriate equivalence relation.
\begin{definition}
	The category \DMix has finite-dimensional Hilbert spaces as its objects. The morphisms of \DMix are generated by elementary morphisms, quotiented by the following equivalence relation $\sim$:
	\begin{align}
		\label{equation:equilitymorphisms}
		\tikzfigscale{0.8}{f} \sim \hspace{1mm}\tikzfigscale{0.8}{g}
		\iff
		&\tikzfigscale{0.8}{f=g1} = \tikzfigscale{0.8}{f=g2} \nonumber\\
		&\text{for all $\mathcal{E}$ and all states $\rho$}
	\end{align}
\end{definition}
Concretely, morphisms in \DMix are arbitrary finite sequences of elementary morphisms (with compatible codomains/domains), quotiented by an equivalence relation enforcing the requirement that two morphisms in \DMix be equal if and only if they act the same when applied to arbitrary subsystems of arbitrary entangled states. It should be noted that, while the generic morphisms in \DMix are abstract objects, the states in \DMix are exactly the normalized quantum states, and the effect of applying a morphism to a state can always be computed following Deutsch's prescription. The definition of \DMix as a process theory and its relationship to ordinary quantum theory are established by the following results (see Appendix \ref{appendix:proofs} for their proof).
\newcounter{theorem:DMixSMC}
\setcounter{theorem:DMixSMC}{\value{theorem_c}}
\begin{theorem}
	The category \DMix is a strict symmetric monoidal category equipped with the following tensor product:
	\begin{equation}
		\tikzfigscale{0.8}{monoidal_product}
		\hspace{4mm}
		:=
		\hspace{3mm}
		\tikzfigscale{0.8}{monoidal_product2}
	\end{equation}
	The symmetry isomorphisms are inherited from quantum theory.
\end{theorem}
\newcounter{lemma:MixDMix}
\setcounter{lemma:MixDMix}{\value{theorem_c}}
\begin{lemma}
	\label{lemma:MixDMix_label}
	The category \Mix of finite-dimensional Hilbert spaces and CPTP maps is faithfully embedded in \DMix by the following monoidal functor:
	\begin{equation}
		\tikzfigscale{0.8}{embedding1new}
		\hspace{3mm}
		\mapsto
		\hspace{3mm}
		\tikzfigscale{0.8}{embedding2new}
	\end{equation}
\end{lemma}

Given the presence of CTCs, and having previously looked at the P-CTC model, one might expect some kind of causal breakdown in the passage from ordinary quantum theory to the D-CTC model. This turns out not to be the case: as the next result shows, the category \DMix respects Relativistic causality constraints (see Appendix \ref{appendix:proofs} for the proof).
\newcounter{theorem:DMixTerminal}
\setcounter{theorem:DMixTerminal}{\value{theorem_c}}
\begin{lemma}
	The category \DMix is terminal (in the parlance of \cite{coecke2013,coecke2016}), i.e. on any given system $\mathcal{H}$ the only effect is the discarding map inherited from ordinary quantum theory:
	\begin{equation}
		\tikzfigscale{0.8}{discarding}
	\end{equation}
	As a consequence, \DMix respects both no-signaling and causality (aka no-signaling from the future).
\end{lemma}
\noindent The fact that Relativistic constraints such as no-signaling and causality are respected in the presence of CTCs is certainly surprising, and is at odds with both folklore and existing literature on the D-CTC model: in the next subsection, we explain this fundamental misunderstanding in terms of a phenomenon known as the ``linearity trap''.

\subsection{The Linearity Trap}

\noindent In the presence of D-CTCs, maps are no longer necessarily linear, nor necessarily continuous: for concrete examples displaying such behaviors, see Appendix \ref{appendix:deutschFeatures}. One of the important consequences of such change in behavior is what \cite{bennett2009} calls the \emph{linearity trap}: this is the derivation of incorrect conclusions from the assumption, implicit in many operational arguments about quantum measurement and preparation, that quantum processes behave linearly, i.e. that they respect classical probabilistic mixtures. The linearity trap is the reason behind some claims in the literature that the D-CTC model allows superluminal signaling---hence violating Relativity---or even worse that it is mathematically inconsistent.

In \cite{brun2009}, the authors describe the following quantum circuit using D-CTCs which allows to perfectly distinguish between the four non-orthogonal states $\ket{\psi} = \ket{0}, \ket{1}, \ket{+}, \ket{-}$:
\begin{equation}
	\label{equation:DCTCdistinguishNonOrthogonal}
	\tikzfigscale{0.8}{circuit_bub}
\end{equation}
In the circuit above, the unitary gates $U_{ij}$ are activated if the two qubits are in the state $\ket{i} \ket{j}$: $U_{00} \equiv SWAP, U_{01} \equiv X \otimes X, U_{10} \equiv (X \otimes I) \circ (H \otimes I), U_{11} \equiv (X \otimes H) \circ (SWAP)$. The circuit yields the following mapping, where each of the two measurements is done in the computational basis and results in a bit of output (so that output is a 2-bit string $ab$):
\begin{align}
	\ket{0} \ket{0} &\mapsto 00,  \ \ \ket{+}\ket{0} \mapsto 10 \nonumber\\
	\ket{1} \ket{0} &\mapsto 01,  \ \ \ket{-}\ket{0} \mapsto 11
\end{align}
The possibility of distinguishing non-orthogonal states is certainly odd, coming from a traditional quantum theoretical perspective, but no more so in our opinion than the absence of linearity or continuity.

In \cite{bub2014}, however, the authors propose that the possibility of distinguishing non-orthogonal states could be used to obtain superluminal signaling. They even venture that the same construction could be used to distinguish between totally mixed states which have been formed by mixture of the Z basis states and by mixture of the X basis states, i.e. between $0.5\ket{0}\bra{0}+0.5\ket{1}\bra{1}$ and $0.5\ket{+}\bra{+}+0.5\ket{-}\bra{-}$: because these are the same exact state, the authors venture that D-CTCs might actually be mathematically inconsistent. In their argument, one considers the following bipartite scenario, where Alice and Bob share an entangled state and $\chi$ is the map described in \ref{equation:DCTCdistinguishNonOrthogonal} (CTC part not shown):
\begin{equation}
	\tikzfigscale{0.8}{discriminate_entanglement}
\end{equation}
By measuring Alice's qubit in the Z basis, the authors claim to obtain the following two scenarios with equal probability:
\begin{equation}
\tikzfigscale{0.8}{scenario_Z1} \ \ \hspace{2mm}\text{ or }\hspace{2mm} \ \ \tikzfigscale{0.8}{scenario_Z2}
\end{equation}
Similarly, by measuring Alice's qubit in the X basis, the authors claim to obtain the following two scenarios with equal probability:
\begin{equation}
\tikzfigscale{0.8}{scenario_X1} \ \ \hspace{2mm}\text{ or }\hspace{2mm} \ \ \tikzfigscale{0.8}{scenario_X2}
\end{equation}
By definition of the $\chi$ map, this would allow Bob to determine whether Alice measured her qubit in the Z basis (output strings $00$ and $01$) or in the X basis (output strings $10$ and $11$)---as well as to learn which outcome she obtained---without any need for communication between them to occur. This would indeed be a violation of no-signaling, but the worrying consequences wouldn't stop at that. In both cases, the state on Bob's left qubit would be the 1-qubit totally mixed state, and the outputs of Bob's measurements would allow him to determine whether the totally mixed state was prepared as $0.5\ket{0}\bra{0}+0.5\ket{1}\bra{1}$ (if Alice measured in the Z basis) or as  $0.5\ket{+}\bra{+}+0.5\ket{-}\bra{-}$ (if Alice measured in the X basis). What is going on here? 

What actually happens in the bipartite scenario above is correctly described by the following diagrams, where the $\Phi$ map is the controlled measurement performed by Alice (measurement in the Z basis for $i=0$, Hadamard followed by measurement in the Z basis for $i=1$):
\begin{equation}
	\tikzfigscale{0.8}{linearitytrap}
	\hspace{3mm}
	=
	\hspace{3mm}
	\tikzfigscale{0.8}{linearitytrap1}
\end{equation}
While it is indeed true that the discarding map on the left can be decomposed as a sum of the two effects $\bra{0} \square \ket{0}$ and $\bra{1} \square \ket{1}$ for the Z basis states, the failure of linearity means that the sum cannot in general be taken out of the dotted box:
\begin{align}
	& \tikzfigscale{0.8}{linearitytrap2} \nonumber \\
	\neq
	\hspace{4mm}
	& \tikzfigscale{0.8}{linearitytrap3}
	\hspace{2mm}
	+
	\hspace{2mm}
	\tikzfigscale{0.8}{linearitytrap4}
\end{align}
This means that it is neither possible to signal, compatibly with the fact that the category \DMix is terminal, nor is it possible to distinguish between the different ways in which the totally mixed state could have been prepared (a fact which would indeed make the model mathematically ill-defined). The lesson we take from the linearity trap is this: when working in \DMix, one should take care to avoid reasoning about quantum measurements by case analysis of individual measurement outcomes.

\section{Categorical semantics for time travel}
\label{section:semantics}

\noindent When causal scenarios in process theories are depicted diagrammatically, it is easy to conflate boxes with processes happening locally at events (i.e. points in spacetime), and wires with the information flow establishing the causal relationships between them. While this practice follows a natural and notationally pleasant convention, it is not mathematically well-founded: except in some specific situations (e.g. the causal categories of \cite{coecke2013}), there need not be a canonical way to decide how a process should be decomposed into sub-processes compatibly with a given causal structure.

When talking about such causal scenarios, one is really combining two distinct ingredients:
\begin{enumerate}
	\item[(i)] a causal graph, representing the events in the scenario and the causal relationship between them;
	\item[(ii)] a map assigning each event in the scenario to the process happening there.
\end{enumerate}
In the usual chronology-respecting (CR) scenarios, where causal graphs are (non-transitive) directed acyclic graphs, the combination of these two ingredients does not pose much of a challenge: a process is associated to each event and outputs/inputs of the process are connected by wires following the edges of the graph (see later for a precise characterization of this procedure). In scenarios involving chronology-violating (CV) regions, however, there is a problem: no prescription exists in a generic symmetric monoidal category for what it means to connect processes in a cycle.

One solution to this problem would be to work in a \emph{compact-closed} symmetric monoidal category: this is exactly what is done in the P-CTC model. A more elegant solution would be to work in a \emph{traced} symmetric monoidal category, with the trace used to close output/input loops: unfortunately, we will soon see that the axioms for a trace are too strong, and that the D-CTC model fails a crucial few of them. To solve this issue, we adopt a more general framework in which the process of applying a CTC is encapsulated by a super-operator satisfying certain trace-like axioms, and we show that this provides sound semantics to the idea of associating processes to events on causal graphs with cycles.

\subsection{Causal graphs}

\noindent The idea of representing discrete scenarios in spacetime by specifying the causal relationship between events is based on a celebrated result by Malament \cite{malament1977}, stating that knowledge of the causal order between all events is enough to reconstruct a past- and future-distinguishing spacetime up to conformal equivalence. This forms the basis of discrete and order-theoretic approaches to Relativity, such as \emph{causal sets} \cite{bombelli1987,rideout1999} and domains \cite{martin2006}.

The events in our scenarios will be discrete, and will correspond to processes happening locally at those events. Even though we will make no such specification, one could take an operational perspective and think of a set of observers well separated in space and/or time, each quickly performing an operation in their local laboratory. Our definition of causal graphs takes inspiration from the definition of causal sets.

\begin{definition}
	A \emph{causal set} $(C, \preceq)$ is a set $C$ endowed with a partial order $\preceq$ which is \emph{locally finite}, i.e. such that for every $x,z \in C$ there are finitely many $y \in C$ such that $x \preceq y \preceq z$.
\end{definition}
\noindent The points of a causal sets represent discrete events in spacetime, with $x \preceq y$ if and only if $x$ causally precedes $y$ in spacetime (i.e. there is a lightlike curve starting at $x$ and ending at $y$). This means that the set $\suchthat{y \in C}{x \preceq y \preceq z}$ is the causal diamond between two events $x$ and $z$, so that the local finiteness condition of causal sets means that every causal diamond in spacetime contains finitely many events in the set.

Our generalization from causal sets to what we will call ``causal graphs'' is due to the need to capture CV scenarios, where the existence of cycles makes the definition as partial order no longer tenable. We will thus take the following equivalent characterization of causal sets as our starting point (see Appendix \ref{appendix:proofs} for the proof).
\newcounter{lemma:causalSetsNontransAcyclicDiraphs}
\setcounter{lemma:causalSetsNontransAcyclicDiraphs}{\value{theorem_c}}
\begin{lemma}
	A causal set is the same as a non-transitive\footnote{i.e. if we $x \rightarrow y$ and $y \rightarrow z$ then $x \not{\hspace{-2.5pt}\rightarrow} z$.} acyclic digraph (directed graph).
\end{lemma}
Generalizing the description above to CV scenarios is simply a matter of dropping the ``acyclic'' requirement. However, we also need to add a framing to our graphs, to deal with the fact that tensor product of objects in symmetric monoidal categories is intrinsically ordered.
\begin{definition}
A \emph{framed causal graph} $\Gamma$ is a non-transitive digraph equipped with the following data:
\begin{itemize}
	\item a sub-set $\inputs{\Gamma}$ of the nodes of $\Gamma$---the \emph{input nodes}---such that each $i \in \inputs{\Gamma}$ has zero incoming edges and a single outgoing edge;
	\item a sub-set $\outputs{\Gamma}$ of the nodes of $\Gamma$---the \emph{output nodes}---such that each $o \in \outputs{\Gamma}$ has zero outgoing edges and a single incoming edge;
	\item a \emph{framing} for $\Gamma$, which consists of the following:
	\begin{itemize}
		\item a total order on $\inputs{\Gamma}$;
		\item a total order on $\outputs{\Gamma}$;
		\item for each node $x \in \Gamma$, a total order on the edges outgoing from $x$, compatible with the total order on $\outputs{\Gamma}$, where relevant;
		\item for each node $x \in \Gamma$, a total order on the edges incoming to $x$, compatible with the total order on $\inputs{\Gamma}$, where relevant.
	\end{itemize}
\end{itemize}
We say that a framed causal graph is \emph{chronology-respecting (CR)} if it is acyclic, and \emph{chronology-violating (CV)} otherwise. Nodes which are neither input nor output are referred to as \emph{internal}, and denoted by $\nodes{\Gamma}$.
\end{definition}

Below are two examples of (finite) framed causal graphs: the left one CR, the right one CV.
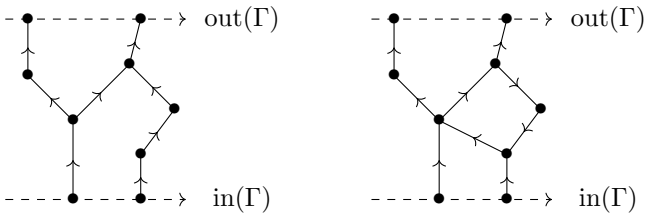
\[
	\begin{tikzpicture}
	\begin{pgfonlayer}{nodelayer}
		\node [style=none] (0) at (-2, 2) {};
		\node [style=none] (1) at (2, 2) {};
		\node [style=none] (2) at (-2, -2) {};
		\node [style=none] (3) at (2, -2) {};
		\node [style=none] (4) at (1.75, 0) {$\bullet$};
		\node [style=none] (5) at (0.75, 1) {$\bullet$};
		\node [style=none] (6) at (1, -1) {$\bullet$};
		\node [style=none] (7) at (1, -2) {$\bullet$};
		\node [style=none] (8) at (-0.5, -2) {$\bullet$};
		\node [style=none] (9) at (1, 2) {$\bullet$};
		\node [style=none] (10) at (-1.5, 0.75) {$\bullet$};
		\node [style=none] (11) at (-0.5, -0.25) {$\bullet$};
		\node [style=none] (12) at (-1.5, 2) {$\bullet$};
		\node [style=none] (13) at (3.25, 2) {$\outputs{\Gamma}$};
		\node [style=none] (14) at (3.25, -2) {$\inputs{\Gamma}$};
	\end{pgfonlayer}
	\begin{pgfonlayer}{edgelayer}
		\draw [->,dashed] (0.center) to (1.center);
		\draw [->,dashed] (2.center) to (3.center);
		\draw [->-] (8.center) to (11.center);
		\draw [->-] (11.center) to (5.center);
		\draw [->-] (11.center) to (10.center);
		\draw [->-] (10.center) to (12.center);
		\draw [->-] (5.center) to (9.center);
		\draw [->-] (7.center) to (6.center);
		\draw [->-] (6.center) to (4.center);
		\draw [->-] (4.center) to (5.center);
	\end{pgfonlayer}
\end{tikzpicture}
	\hspace{1.2cm}
	\begin{tikzpicture}
	\begin{pgfonlayer}{nodelayer}
		\node [style=none] (0) at (-2, 2) {};
		\node [style=none] (1) at (2, 2) {};
		\node [style=none] (2) at (-2, -2) {};
		\node [style=none] (3) at (2, -2) {};
		\node [style=none] (4) at (1.75, 0) {$\bullet$};
		\node [style=none] (5) at (0.75, 1) {$\bullet$};
		\node [style=none] (6) at (1, -1) {$\bullet$};
		\node [style=none] (7) at (1, -2) {$\bullet$};
		\node [style=none] (8) at (-0.5, -2) {$\bullet$};
		\node [style=none] (9) at (1, 2) {$\bullet$};
		\node [style=none] (10) at (-1.5, 0.75) {$\bullet$};
		\node [style=none] (11) at (-0.5, -0.25) {$\bullet$};
		\node [style=none] (12) at (-1.5, 2) {$\bullet$};
		\node [style=none] (13) at (3.25, 2) {$\outputs{\Gamma}$};
		\node [style=none] (14) at (3.25, -2) {$\inputs{\Gamma}$};
	\end{pgfonlayer}
	\begin{pgfonlayer}{edgelayer}
		\draw [->,dashed] (0.center) to (1.center);
		\draw [->,dashed] (2.center) to (3.center);
		\draw [->-] (8.center) to (11.center);
		\draw [->-] (11.center) to (5.center);
		\draw [->-] (11.center) to (10.center);
		\draw [->-] (10.center) to (12.center);
		\draw [->-] (5.center) to (9.center);
		\draw [->-] (7.center) to (6.center);
		\draw [->-] (4.center) to (6.center);
		\draw [->-] (5.center) to (4.center);
		\draw [->-] (6.center) to (11.center);
	\end{pgfonlayer}
\end{tikzpicture}
\]
Framed causal graphs form a symmetric monoidal category, with natural numbers as objects and framed causal graphs $G$ with $\#\inputs{G} = n$ and $\#\outputs{G} = m$ as morphisms $n \rightarrow m$. Sequential composition $H \circ G$ is done by gluing $\outputs{G}$ with $\inputs{H}$, while parallel composition $H \oplus G$ is done by stacking the two graphs side by side, with the sum total orders of inputs and outputs. See Appendix \ref{appendix:framedCausalGraphsSMC} for the full details of the construction.
\begin{definition}
	A \emph{diagram over a causal graphs $\Gamma$} in a process theory $\textbf{C}$ is a pair of maps $(\alpha,\beta)$ as follow:
	\begin{itemize}
		\item a map $\alpha: \edges{\Gamma} \rightarrow \obj{\textbf{C}}$ assigning an object of $\textbf{C}$ to each edge of the graph;
		\item a map $\beta: \nodes{\Gamma} \rightarrow \mor{\textbf{C}}$ assigning a morphism of $\textbf{C}$ to each internal node of the graph, such that if $(e_1, \ldots, e_n)$ are the ordered incoming edges and $(f_1, \ldots, f_m)$ are the ordered outgoing edges of a node $v$, then the morphism $\beta(v)$ is of type:
		\[
			\beta(v): \bigotimes_i \alpha(e_i) \rightarrow \bigotimes_j \alpha(f_j)
		\]
	\end{itemize}
\end{definition}

\subsection{Diagrams over CR causal graphs}

\noindent It is really straightforward to interpret diagrams over CR causal graphs as morphisms. From here on, we assume without loss of generality that the process theory $\textbf{C}$ under consideration is a strict symmetric monoidal category.
\begin{definition}
\label{definition:morphismsOverCR}
	Let $(\alpha,\beta)$ be a diagram over a CR causal graph $\Gamma$ in a process theory $\textbf{C}$. The \emph{morphism defined by the diagram $(\alpha,\beta)$} is the morphism
	\begin{equation}
		\Phi_{(\alpha,\beta)}: \bigotimes_{i \in \inputs{\Gamma}} \alpha(!_i) \rightarrow \bigotimes_{o \in \outputs{\Gamma}} \alpha(!_o),
	\end{equation}
	(where $!_i$ is the unique edge outgoing from an input node $i$ and $!_o$ is the unique edge incoming into an output node $o$) defined in the obvious way by joining outputs and inputs of the processes specified by $\beta$ along the edges of the graph $\Gamma$.
\end{definition}

Below is an example associating four processes $f$, $g$, $h$ and $k$ to the four internal nodes of a CR causal graph:
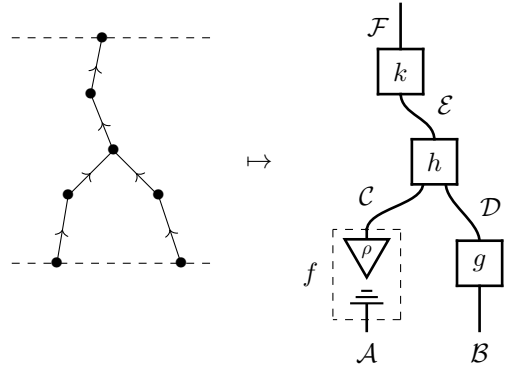
\[
	\begin{tikzpicture}
	\begin{pgfonlayer}{nodelayer}
		\node [style=none] (0) at (-1, -0.75) {$\bullet$};
		\node [style=none] (1) at (-0.5, 1.5) {$\bullet$};
		\node [style=none] (2) at (1, -0.75) {$\bullet$};
		\node [style=none] (3) at (0, 0.25) {$\bullet$};
		\node [style=none] (4) at (0, 0.25) {};
		\node [style=none] (5) at (1.5, -2.25) {$\bullet$};
		\node [style=none] (6) at (-1.25, -2.25) {$\bullet$};
		\node [style=none] (7) at (-0.25, 2.75) {$\bullet$};
		\node [style=none] (8) at (-2.25, 2.75) {};
		\node [style=none] (9) at (2.25, 2.75) {};
		\node [style=none] (10) at (-2.25, -2.25) {};
		\node [style=none] (11) at (2.25, -2.25) {};
	\end{pgfonlayer}
	\begin{pgfonlayer}{edgelayer}
		\draw [style=->-] (2.center) to (3.center);
		\draw [style=->-] (0.center) to (4.center);
		\draw [style=->-] (3.center) to (1.center);
		\draw [style=->-] (5.center) to (2.center);
		\draw [style=->-] (6.center) to (0.center);
		\draw [style=->-] (1.center) to (7.center);
		\draw [style=dashed] (8.center) to (9.center);
		\draw [style=dashed] (10.center) to (11.center);
	\end{pgfonlayer}
\end{tikzpicture}
	\hspace{3mm}
	\mapsto
	\hspace{3mm}
	\begin{tikzpicture}
	\begin{pgfonlayer}{nodelayer}
		\node [style=dpoint] (0) at (-1, -2) {$\rho$};
		\node [style=none] (1) at (-0.25, 2) {$k$};
		\node [style=none] (2) at (-0.25, 2.5) {};
		\node [style=none] (3) at (0.25, -0.5) {};
		\node [style=none] (4) at (-0.25, 1.5) {};
		\node [style=none] (5) at (1.5, -2.25) {$g$};
		\node [style=none] (6) at (0, -0.5) {};
		\node [style=none] (8) at (1, -1.75) {};
		\node [style=none] (9) at (-0.25, 2.5) {};
		\node [style=none] (10) at (0.5, 0) {$h$};
		\node [style=none] (11) at (0.5, 0) {};
		\node [style=none] (12) at (0.75, -0.5) {};
		\node [style=none] (14) at (0.5, 0.5) {};
		\node [style=none] (15) at (-1, -1.5) {};
		\node [style=none] (16) at (1, -2.75) {};
		\node [style=none] (17) at (-0.25, -3.5) {};
		\node [style=none] (18) at (0.5, 0.5) {};
		\node [style=none] (19) at (0.25, 1.5) {};
		\node [style=none] (20) at (0, 0.5) {};
		\node [style=none] (21) at (-0.75, 1.5) {};
		\node [style=none] (23) at (-0.25, 3.5) {};
		\node [style=none] (25) at (-0.25, 1.5) {};
		\node [style=none] (26) at (1, -0.5) {};
		\node [style=none] (28) at (0.25, 2.5) {};
		\node [style=none] (29) at (2, -2.75) {};
		\node [style=none] (30) at (-0.25, -1.5) {};
		\node [style=none] (31) at (1.5, -1.75) {};
		\node [style=none] (32) at (-1.75, -3.5) {};
		\node [style=none] (33) at (2, -1.75) {};
		\node [style=none] (34) at (1, 0.5) {};
		\node [style=none] (35) at (1.5, -3.75) {};
		\node [style=none] (36) at (-0.75, 2.5) {};
		\node [style=none] (37) at (-2.25, -2.5) {$f$};
		\node [style=none] (38) at (1.5, -2.75) {};
		\node [style=none] (39) at (-0.25, 1.5) {};
		\node [style=none] (40) at (-1.75, -1.5) {};
		\node [style=upground] (41) at (-1, -3) {};
		\node [style=none] (42) at (-1, -3.75) {};
		\node [style=none] (43) at (-1, -4.25) {$\mathcal{A}$};
		\node [style=none] (44) at (1.5, -4.25) {$\mathcal{B}$};
		\node [style=none] (45) at (-1, -0.75) {$\mathcal{C}$};
		\node [style=none] (46) at (1.75, -1) {$\mathcal{D}$};
		\node [style=none] (47) at (0.75, 1.25) {$\mathcal{E}$};
		\node [style=none] (48) at (-0.75, 3) {$\mathcal{F}$};
	\end{pgfonlayer}
	\begin{pgfonlayer}{edgelayer}
		\draw [style=boldedge] (20.center) to (6.center);
		\draw [style=boldedge] (26.center) to (6.center);
		\draw [style=boldedge] (26.center) to (34.center);
		\draw [style=boldedge] (34.center) to (20.center);
		\draw [style=boldedge] (36.center) to (21.center);
		\draw [style=boldedge] (19.center) to (21.center);
		\draw [style=boldedge] (19.center) to (28.center);
		\draw [style=boldedge] (28.center) to (36.center);
		\draw [style=boldedge] (8.center) to (16.center);
		\draw [style=boldedge] (29.center) to (16.center);
		\draw [style=boldedge] (29.center) to (33.center);
		\draw [style=boldedge] (33.center) to (8.center);
		\draw [style=dashed] (40.center) to (32.center);
		\draw [style=dashed] (17.center) to (32.center);
		\draw [style=dashed] (17.center) to (30.center);
		\draw [style=dashed] (30.center) to (40.center);
		\draw [style=boldedge] (35.center) to (38.center);
		\draw [style=boldedge] (0) to (15.center);
		\draw [style=boldedge] (25.center) to (4.center);
		\draw [style=boldedge] (9.center) to (23.center);
		\draw [style=boldedge, in=90, out=-90] (25.center) to (14.center);
		\draw [style=boldedge, in=-90, out=90, looseness=0.75] (15.center) to (3.center);
		\draw [style=boldedge, in=90, out=-90, looseness=0.75] (12.center) to (31.center);
		\draw [style=boldedge] (42.center) to (41);
	\end{pgfonlayer}
\end{tikzpicture}
\]
In this example, the process associated to the diagram over the CR causal graph is as follows:
\[
	k \circ h \circ (f \otimes g): \mathcal{A} \otimes \mathcal{B} \rightarrow \mathcal{F}
\]
The diagram above also exemplifies how merely looking at the boxes themselves may be deceitful: box $f$ factorizes, but this has nothing to do with causal relations between events  (which are captured by the causal graph instead).

\subsection{Diagrams over CV-local causal graphs}

\noindent We begin by considering a special class of CV framed causal graphs, which we refer to as \emph{CV-local}: CTCs are allowed to appear, but we want the interaction between the CV part and the CR part to remain completely localized to a discrete set of events (in a sense made precise below). This requirement captures our original stance about time-travel and CTCs, namely that interaction between the CR and CV regions of spacetime must happen in a fully localized fashion.
\begin{definition}
	A framed causal graph is said to be \emph{CV-local} if the following condition holds: every simple cycle, by which we mean a cycle which is both edge-disjoint and vertex-disjoint, has at most one node of degree higher than $2$. Simple cycles are then referred to as the \emph{CTCs} and the unique node of degree higher than 2 on a simple cycle is referred to as the \emph{interaction node} for the CTC.
\end{definition}
\noindent CV-local framed causal graphs are closed under composition and tensor product, and hence form a monoidal sub-category of the SMC of framed causal graphs. Below is an example of CV-local graph (interaction node highlighted in red):
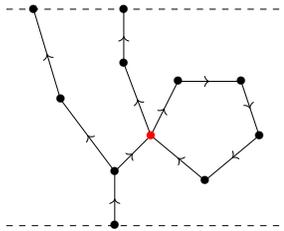
\[
	\tikzfigscale{0.8}{CVlocal}
\]
which decomposes into the following CR and CV regions (overlapping at the interaction node, again highlighted in red in both figures):
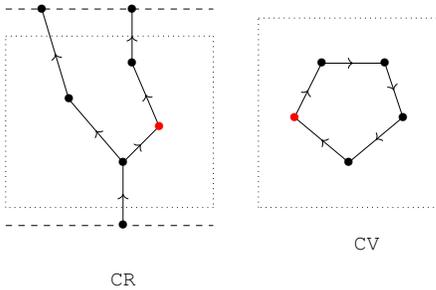
\[
	\tikzfigscale{0.8}{CVlocalDecomposition}
\]
Below is instead an example of a framed causal graph which is not CV-local, with the extended overlap between the CR and CV regions highlighted by the red dashed box:
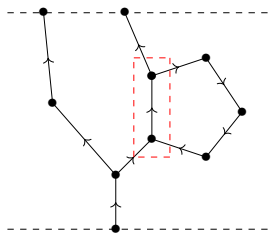
\[
	\tikzfigscale{0.8}{not_CVlocal}
\]
The graph above clearly exemplifies why CV-locality is needed: across the dotted box, systems coming from the CR region must all enter the CTC before returning to the CR region, something which our interpretation of time-travel states not to be meaningful in the first place.

Our first task is to define what it means to have diagrams over CV-local framed causal graphs. In the P-CTC model, this is extremely simple:
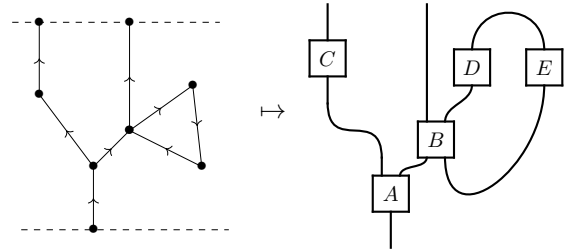
\[
	\tikzfigscale{0.8}{CVgraph}
	\hspace{2mm}
	\mapsto
	\hspace{2mm}
	\tikzfigscale{0.8}{CVgraphMorphism}
\]
By adjusting the boxes a little, Lloyd-like semantics can be readily extended to traced monoidal categories. However, we will soon see that the axioms for traces are too restrictive to capture the model we most care about, i.e. the D-CTC model. Instead, we observe that CV-local framed causal graphs can be transformed into CR ones by ``cutting the cycles open'':
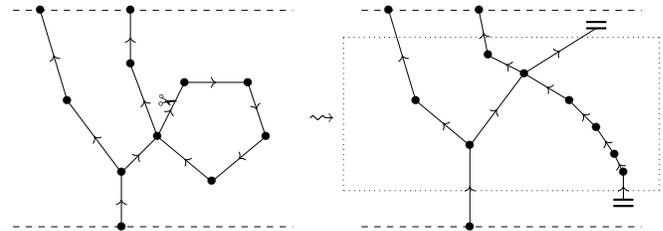
\[
	\tikzfigscale{0.8}{CVlocalS} \rightsquigarrow \tikzfigscale{0.8}{CVlocal_cut}
\]
In case of multiple cycles meeting with the CR region at the same interaction node, each cycle is cut open independently:
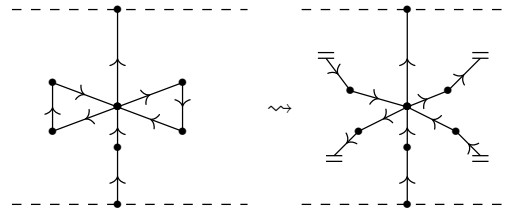
\[
	\tikzfigscale{1.2}{two_cycles_cut} \rightsquigarrow \tikzfigscale{1.2}{two_cycles}
\]
This means that we can decompose the task of drawing diagrams over CV-local graphs into two distinct parts: drawing diagrams over CR graphs, which we already known how to do, and finding a way to glue the cut ends of the cycles back together at the end. The latter part will be carried out by a ``time-travel super-operator'', taking a pair of input/output wires and ``applying a CTC'' to them. Because there is no one canonical place to cut cycles open at, such a ``time-travel super-operator'' will have to satisfy a certain ``sliding property'', to ensure that all ways of cutting a cycle open will lead to the same diagram in the end:
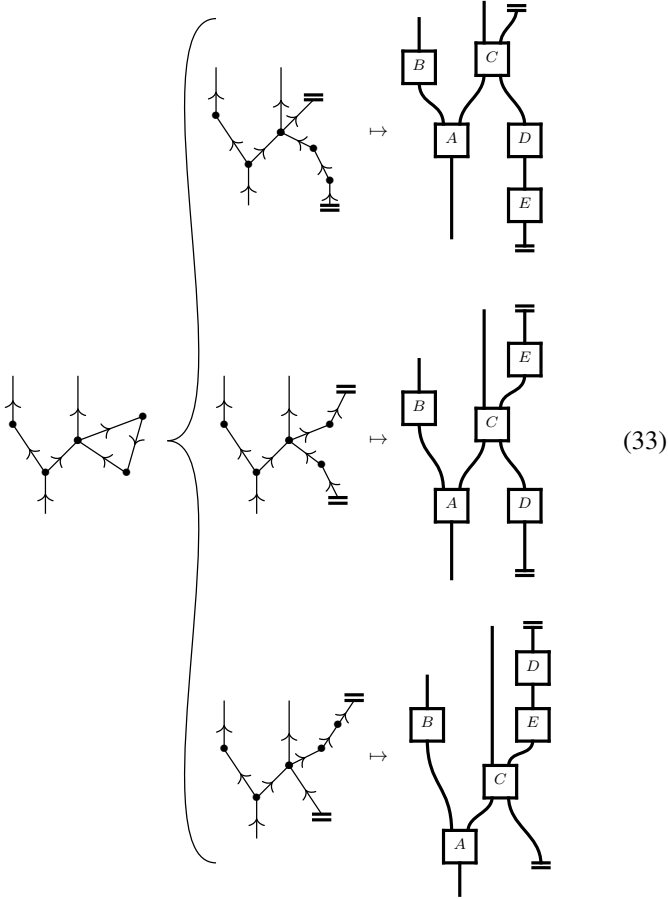
\begin{equation}
	\hspace{-1.35cm}\tikzfigscale{1.2}{big_graphs}
\end{equation}

\begin{definition}
	\label{def:time-travel}
	A \emph{time-travel super-operator} on a symmetric monoidal category $\textbf{C}$ is a pair of a symmetric monoidal category $\textbf{D}$, into which $\textbf{C}$ is monoidally faithfully embedded, together with a family of super-operators $\Xi_{\mathcal{H};\mathcal{C}}^{\mathcal{K};\mathcal{C}}$:
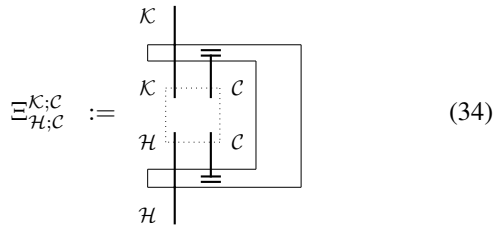
	\begin{equation}
		\Xi_{\mathcal{H};\mathcal{C}}^{\mathcal{K};\mathcal{C}}
		\hspace{2mm}
		:=
		\hspace{2mm}\tikzfigscale{0.8}{timetravelsuperoperator}
	\end{equation}
	taking morphisms $\Phi: \mathcal{H} \otimes \mathcal{C} \rightarrow \mathcal{K} \otimes \mathcal{C}$ of $\textbf{C}$ to morphisms $\Xi_{\mathcal{H};\mathcal{C}}^{\mathcal{K};\mathcal{C}}\left(\Phi\right): \mathcal{H} \rightarrow \mathcal{K}$ of $\textbf{D}$ and satisfying the following:
	\begin{enumerate}
		\item Naturality in the CR region:
		\begin{equation}
			\tikzfigscale{0.8}{naturalityCR} \hspace{2mm}=\hspace{2mm} \tikzfigscale{0.8}{naturalityCR2}
		\end{equation}
		\item Strength:
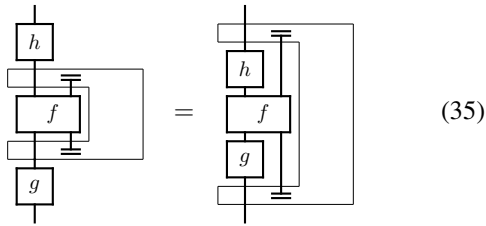
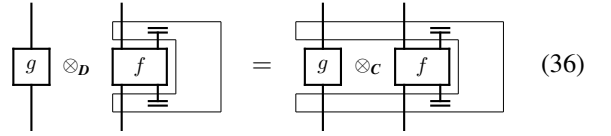
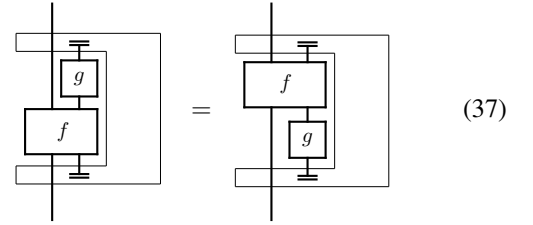
		\begin{equation}
			\tikzfigscale{0.8}{superposing} \hspace{2mm}=\hspace{2mm}\tikzfigscale{0.8}{superposing2}
		\end{equation}
		\item Sliding:
		\begin{equation}
		\label{equation:slidingproperty}
			\tikzfigscale{0.8}{sliding1} \hspace{2mm}=\hspace{2mm} \tikzfigscale{0.8}{sliding2}
		\end{equation}
		\item Vanishing (dashed line is the tensor unit):
		\begin{equation}
		\label{equation:vanishingproperty}
			\tikzfigscale{0.8}{vanishing1} \hspace{2mm}=\hspace{2mm} \tikzfigscale{0.8}{vanishing2}
		\end{equation}
		\end{enumerate}
		We refer to a process theory equipped with a time-travel super-operator as a \emph{process theory with time travel}.
\end{definition}

We now present our two main results about time-travel super-operators: the first one shows that the definition captures the D-CTC and the P-CTC models, while the second one shows that the definition is necessary and sufficient to provide sound semantics to diagrams over CV-local framed causal graphs in our ``cut-and-paste'' approach (see Appendix \ref{appendix:proofs} for the proofs).

\newcounter{theorem:DMixProperties}
\setcounter{theorem:DMixProperties}{\value{theorem_c}}
\begin{theorem}
	\label{theorem:DMixProperties_label}
	D-CTCs and the category $\DMix$ define a time-travel super-operator for quantum theory, i.e. for the symmetric monoidal category $\Mix$ of finite-dimensional Hilbert spaces and CPTP maps. The same is true of P-CTCs and the category $\Mix_{sym}$.
\end{theorem}

\begin{definition}
\label{definition:morphismsOverCVlocal}
	Let $(\alpha,\beta)$ be a diagram over a CV-local causal graph $\Gamma$ in a process theory $\textbf{C}$ with a time-traveling super-operator. The \emph{morphism defined by the diagram $(\alpha,\beta)$} is the morphism obtained with the following procedure:
	\begin{enumerate}
			\item We cut every CTC open at some edge on the cycle, and the two cut ends of the edge are marked to remember that they need to be glued again. Each pair of cut ends forms a pair of an input node and an output node for the resulting CR framed causal graph, with the same system associated to them by $\alpha$.
			\item We consider the morphism defined by diagram $(\alpha,\beta)$ over the resulting CR framed causal graph.
			\item We apply the time-travel super-operator to each pair of input/output nodes that resulted from cutting a cycle open (and hence needs to be ``glued back together''). If an interaction node had more than one cycle through it, all pairs of input/output nodes for those cycles have to be glued together by the same super-operator application, as in the example below:
			\begin{equation}
				\tikzfigscale{0.8}{example}
				\mapsto
				\hspace{2mm}
				\tikzfigscale{0.5}{example2}
			\end{equation}
	\end{enumerate}
\end{definition}

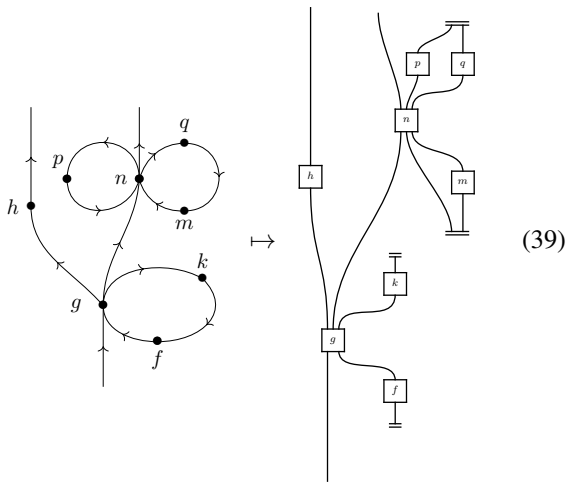
\newcounter{theorem:CVlocalSemantics}
\setcounter{theorem:CVlocalSemantics}{\value{theorem_c}}
\begin{theorem}
	In the presence of a time-travel super-operator, morphisms defined by diagrams over CV-local framed causal graphs are well-defined. Conversely, any super-operator which yields well-defined morphisms for diagrams over CV-local framed causal graphs (and respects the embedding of the original process theory) must satisfy the properties of a time-traveling super-operator.
\end{theorem}

\subsection{Time-travel super-operators vs. traces}

\noindent We have seen that the definition of the time-travel super-operator is general enough to capture both the D-CTC model and the P-CTC model, and at the same time specific enough to provide sound semantics for diagrams over CV-local framed causal graphs. The experienced reader will have already noticed that the defining properties of time-travel super-operator mimic a subset of the defining properties for traces \cite{selinger2011}. In particular, the P-CTC part of Theorem \ref{theorem:DMixProperties_label} is a consequence of the fact that in a traced monoidal category $\textbf{C}$ the trace itself provides the time-travel super-operators, with super-category $\textbf{D}$ equal to $\textbf{C}$ itself.

So why did we not adopt traced monoidal categories as our semantics for time-travel? The answer lies in the following \emph{yanking property}, part of the definition of traces which Inequality \ref{equation:failureOfTomography} tells us not to hold for the D-CTC model:
\begin{equation}
\label{equation:yanking}
	\tikzfigscale{0.8}{unitality}
\end{equation}
The next result shows that failure to satisfy the equation above is an obstruction to providing sound semantics for diagrams over arbitrary framed causal graphs. This in turn implies something interesting about the D-CTC model: that it is really a model about local interaction with CTCs, rather than a model about full-fledged time travel.
\newcounter{theorem:DeutschCannotDoAllCVgraphs}
\setcounter{theorem:DeutschCannotDoAllCVgraphs}{\value{theorem_c}}
\begin{theorem}
\label{theorem:DeutschCannotDoAllCVgraphs_label}
	Consider a process theory with discarding maps and time travel. If the time-travel super-operator yields well-defined morphisms for diagrams over arbitrary framed causal graphs, then it must satisfy the yanking property.
\end{theorem}

We conclude with the following summary of the differences between traces---well understood but limited to Lloyd-like models of time-travel---and the time-travel super-operators defined in this work---capturing more powerful time-travel models such as the D-CTC model of Deutsch.
\begin{itemize}
	\item Traces satisfy the yanking property from equation \ref{equation:yanking}, while time-travel super-operators are not required to (the D-CTC model, in particular, fails it).
	\item Traces are defined on all morphisms of the category, while time-travel super-operators send morphisms of the original category into a larger category, over which they are not necessarily defined.
	\item As a consequence of the point above, the following property of traces---the last remaining one required for their definition---cannot even be formulated for time-travel super-operators, as it needs an application of the super-operator to a morphism living in the super-category \textbf{D}:
	\begin{equation}
		\tikzfigscale{0.8}{undefined} =\tikzfigscale{0.8}{undefined2}
	\end{equation}
\end{itemize}

\section*{Acknowledgments}
We thank Seth Lloyd for discussions on CTCs.
NP gratefully acknowledges funding from EPSRC and the Pirie-Reid Scholarship.
This publication was made possible through the support of a grant on Quantum Causal Structures from the John Templeton Foundation. The opinions expressed in this publication are those of the authors and do not necessarily reflect the views of the John Templeton Foundation. 



\newpage
\appendices

\section{Graphical calculus}
\label{appendix_graphical}

\noindent The symmetric monoidal category that we will treat as a framework for quantum theory is the category $\Mix$, whose objects are finite dimensional Hilbert spaces and morphisms are completely positive trace preserving (CPTP) maps (i.e. we work in a subcategory of $\mathbf{CPM}(\mathbf{FdHilb})$ \cite{selinger2007,coecke2008}). The morphisms of this category will be denoted as boxes, with inputs at the bottom and outputs at the top: 
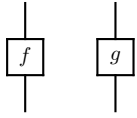
\[
\tikzfigscale{0.8}{doubledboxes}
\]
Working in this category, we make use of the ZX-calculus \cite{coecke2011,backens2014,ng2017,coecke_kissinger_book} and of the graphical notation for symmetric monoidal categories provided by string diagrams \cite{selinger2011}. We will not give an introduction to string diagrams and the ZX-calculus, both of which can be found in \cite{coecke_kissinger_book}, but we clarify here some notational aspects relevant to this work.

In the category $\Mix$, morphisms from the monoidal unit are density matrices, i.e. quantum states. When working with qubits, we denote the computational $Z$ basis states $\{\ket{0}\bra{0},\ket{1}\bra{0}\}$ in white as follows:
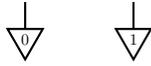
\[
	\tikzfigscale{0.8}{whitestates}
\]
It is worth noting that each $Z$ basis state can also be written as a phase state for the $X$ basis, which we denote in gray as follows:
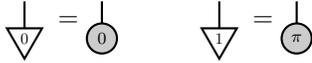
\[
	\tikzfigscale{0.8}{zerowhite} =\tikzfigscale{0.8}{0gray}
	\hspace{1cm}
	\tikzfigscale{0.8}{onewhite} = \tikzfigscale{0.8}{pigray}
\]
Phases are a special case of spiders. The \emph{$Z$ spider} \cite{coecke_kissinger_book} with phase $\alpha$ is the completely positive map associated to the following Hilbert space map:
\[
	\ket{0 \ldots 0}\bra{0\ldots 0} + e^{i \alpha} \ket{1 \ldots 1}\bra{1\ldots 1} 
\]
Similarly, the \emph{$X$ spider} with phase $\alpha$ is the completely positive map associated to the following Hilbert space map:
\[
	\ket{+ \ldots +}\bra{+\ldots +} + e^{i \alpha} \ket{- \ldots -}\bra{-\ldots -} 
\]
In the graphical language, the $Z$ and $X$ spiders are denoted as follows, in white and gray respectively:
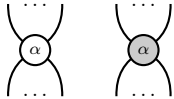
\[
	\tikzfigscale{0.8}{spiders}
\]
When the phase is $\alpha = 0$, we often omit it altogether in the graphical notation. Thus, the CNOT gate is denoted as follows:
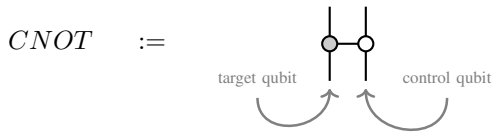
\begin{equation*}
	CNOT
	\hspace{5mm}:=\hspace{5mm}
	\tikzfigscale{0.8}{cnot}
\end{equation*}
The NOT gate, i.e. the Pauli $X$ gate, is denoted as follows:
\begin{equation*}
	NOT
	\hspace{5mm}:=\hspace{5mm}
	\tikzfigscale{0.8}{not}
\end{equation*}
For reasons of notational convenience, we choose the normalization for \graydot in such a way that the CNOT gate can be written without additional normalization scalars. The density matrix for maximally entangled Bell state is denoted by the \emph{cup}:
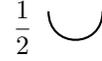
\[
	\frac{1}{2}\hspace{2mm}\tikzfigscale{0.8}{cup}
\]
The unique deterministic effect, aka the \emph{discarding map}, is denoted as follows:
\[
	\tikzfigscale{0.8}{causality2}
\]
The requirement that CP maps be trace-preserving is captured by the following equation, satisfied by all morphisms in $\Mix$:
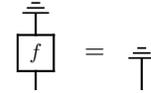
\[
	\tikzfigscale{0.8}{causality1}
	\hspace{2mm} = \hspace{2mm}
	\tikzfigscale{0.8}{causality2}
\]
A process theory in which the equation above is satisfied by all morphisms is known as \emph{terminal}. Terminality of a process theory has been shown to be the same as satisfying the Relativistic constraints of causality and no-signaling \cite{coecke2016}.

\section{Fixed point of maximal entropy}
\label{appendix:fixedPoint}

\noindent Let $\Phi: \mathcal{H} \rightarrow \mathcal{H}$ be a be a CPTP map on a $d$-dimensional Hilbert space $\mathcal{H}$. Let $K := \left\{ \rho \mid \vert \Phi(\rho) = \rho \right\}$ be the convex set of normalized states of $\mathcal{H}$ fixed by $\Phi$. Let $G < U(d)/U(1)$ be the maximal subgroup sending $K$ to $K$ in the conjugation action:
\[
	u(\rho) := u \rho u^\dagger
\]
A complete characterization of the set of fixed points of a CPTP map can be found in \cite{blumekohout2008}. It can be shown that the unique fixed point under the action of $G$, i.e. the unique normalized state $\tau \in K$ such that $u(\tau)=\tau$ for all $u \in G$, is the average obtained by taking the following orbital integral for any $\sigma \in K$:
\[
	\tau := \int_G u(\sigma) \  du
\]
If now we take any $\sigma \in K$, we can use the concavity of the von Neumann entropy functional $S$ to deduce that:
\[
	S(\tau) = S\left(\int_G u(\sigma) \ du \right) \geq \int_G S(u(\sigma)) \ du
\]
And we can further use invariance of $S$ under the unitary conjugation action to deduce that:
\[
	\int_G S(u(\sigma)) \ du = \int_G S(\sigma) \ du = S(\sigma) \int_G du = S(\sigma)
\]
By putting the two together, we conclude that $S(\tau) \geq S(\sigma)$, where $\sigma$ was arbitrary, so that $\tau$ is really a fixed point of maximal entropy. The structure of the set of fixed point proven in \cite{blumekohout2008} can furthermore be used to conclude that $\tau$ is the \emph{unique} fixed point of maximal entropy.

\section{Weird features of the the D-CTC model}
\label{appendix:deutschFeatures}

\subsection{Nonlinearity}

\noindent To show that processes in the D-CTC model can be nonlinear, we consider the map used before in the grandfather's paradox, but swapping Z and X in the CNOT gate:
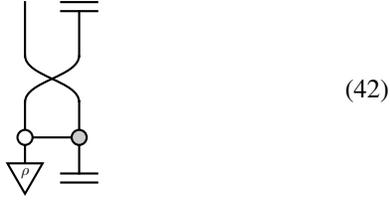
\begin{equation}
	\label{equation:timetravelTONCrho}
	\tikzfigscale{0.8}{timetravelTONCrho}
\end{equation}
The fixed point for input $\rho = \ket{1}\bra{1}$ is computed to be $\tau = \ket{1}\bra{1}$ as follows:
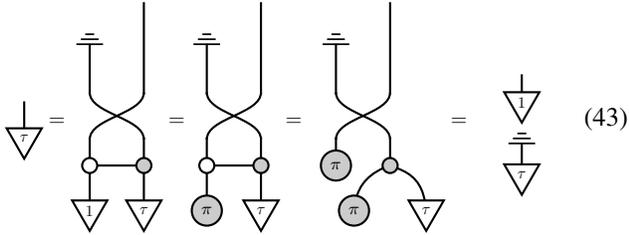
\begin{equation}
	\scalebox{0.8}{$
		\begin{tikzpicture}
	\begin{pgfonlayer}{nodelayer}
		\node [style=dpoint] (0) at (0, -0.5) {$\tau$};
		\node [style=none] (1) at (0, 0.5) {};
	\end{pgfonlayer}
	\begin{pgfonlayer}{edgelayer}
		\draw [style=boldedge] (1.center) to (0);
	\end{pgfonlayer}
\end{tikzpicture} = \begin{tikzpicture}
	\begin{pgfonlayer}{nodelayer}
		\node [style=upground] (0) at (-0.75, 2.25) {};
		\node [style=white ddot] (1) at (-0.75, -1.25) {};
		\node [style=dpoint] (2) at (0.75, -2.5) {$\tau$};
		\node [style=none] (3) at (-0.75, 0.75) {};
		\node [style=none] (4) at (0.75, -0.5) {};
		\node [style=none] (5) at (0.75, 0.75) {};
		\node [style=none] (6) at (-0.75, -0.5) {};
		\node [style=gray ddot] (7) at (0.75, -1.25) {};
		\node [style=dpoint] (8) at (-0.75, -2.5) {$1$};
		\node [style=none] (9) at (0.75, 3.25) {};
	\end{pgfonlayer}
	\begin{pgfonlayer}{edgelayer}
		\draw [style=boldedge, in=-90, out=90, looseness=0.75] (6.center) to (5.center);
		\draw [style=boldedge, in=90, out=-90, looseness=0.75] (3.center) to (4.center);
		\draw [style=boldedge] (5.center) to (9.center);
		\draw [style=boldedge] (4.center) to (7);
		\draw [style=boldedge] (7) to (2);
		\draw [style=boldedge] (3.center) to (0);
		\draw [style=boldedge] (6.center) to (1);
		\draw [style=boldedge] (1) to (7);
		\draw [style=boldedge] (1) to (8);
	\end{pgfonlayer}
\end{tikzpicture}
		=
		\begin{tikzpicture}
	\begin{pgfonlayer}{nodelayer}
		\node [style=gray ddot] (0) at (0.75, -1.25) {};
		\node [style=white ddot] (1) at (-0.75, -1.25) {};
		\node [style=none] (2) at (0.75, 3.25) {};
		\node [style=none] (3) at (-0.75, 0.75) {};
		\node [style=dpoint] (4) at (0.75, -2.5) {$\tau$};
		\node [style=gray phase ddot] (5) at (-0.75, -2.5) {$\pi$};
		\node [style=none] (6) at (0.75, -0.5) {};
		\node [style=upground] (7) at (-0.75, 2.25) {};
		\node [style=none] (8) at (0.75, 0.75) {};
		\node [style=none] (9) at (-0.75, -0.5) {};
	\end{pgfonlayer}
	\begin{pgfonlayer}{edgelayer}
		\draw [style=boldedge, in=-90, out=90, looseness=0.75] (9.center) to (8.center);
		\draw [style=boldedge, in=90, out=-90, looseness=0.75] (3.center) to (6.center);
		\draw [style=boldedge] (8.center) to (2.center);
		\draw [style=boldedge] (6.center) to (0);
		\draw [style=boldedge] (0) to (4);
		\draw [style=boldedge] (3.center) to (7);
		\draw [style=boldedge] (9.center) to (1);
		\draw [style=boldedge] (1) to (0);
		\draw [style=boldedge] (1) to (5);
	\end{pgfonlayer}
\end{tikzpicture}
		=
		\hspace{2mm}
		\begin{tikzpicture}
	\begin{pgfonlayer}{nodelayer}
		\node [style=none] (0) at (-0.75, -0.5) {};
		\node [style=gray phase ddot] (1) at (-0.75, -1.25) {$\pi$};
		\node [style=none] (2) at (-0.75, 0.75) {};
		\node [style=dpoint] (3) at (1.75, -2.5) {$\tau$};
		\node [style=upground] (4) at (-0.75, 2.25) {};
		\node [style=none] (5) at (0.75, 3.25) {};
		\node [style=gray ddot] (6) at (0.75, -1.25) {};
		\node [style=none] (7) at (0.75, -0.5) {};
		\node [style=none] (8) at (0.75, 0.75) {};
		\node [style=gray phase ddot] (9) at (-0.25, -2.5) {$\pi$};
	\end{pgfonlayer}
	\begin{pgfonlayer}{edgelayer}
		\draw [style=boldedge, in=-90, out=90, looseness=0.75] (0.center) to (8.center);
		\draw [style=boldedge, in=90, out=-90, looseness=0.75] (2.center) to (7.center);
		\draw [style=boldedge] (8.center) to (5.center);
		\draw [style=boldedge] (7.center) to (6);
		\draw [style=boldedge, bend left, looseness=1.00] (6) to (3);
		\draw [style=boldedge] (2.center) to (4);
		\draw [style=boldedge] (0.center) to (1);
		\draw [style=boldedge, bend left, looseness=1.00] (9) to (6);
	\end{pgfonlayer}
\end{tikzpicture}
		=
		\hspace{5mm}
		\begin{tikzpicture}
	\begin{pgfonlayer}{nodelayer}
		\node [style=upground] (0) at (0, -0.5) {};
		\node [style=dpoint] (1) at (0, 0.5) {$1$};
		\node [style=dpoint] (2) at (0, -1.5) {$\tau$};
		\node [style=none] (3) at (0, 1.25) {};
	\end{pgfonlayer}
	\begin{pgfonlayer}{edgelayer}
		\draw [style=boldedge] (3.center) to (1);
		\draw [style=boldedge] (2) to (0);
	\end{pgfonlayer}
\end{tikzpicture}
	$}
\end{equation}
Similarly, fixed point for input $\rho = \ket{0}\bra{0}$ is computed to be $\tau = \ket{0}\bra{0}$. Using these results, one can show that the D-CTC map given in \ref{equation:timetravelTONCrho} above sends both input states $\ket{0}\bra{0}$ and $\ket{1}\bra{1}$ to the output state $\ket{0}\bra{0}$:
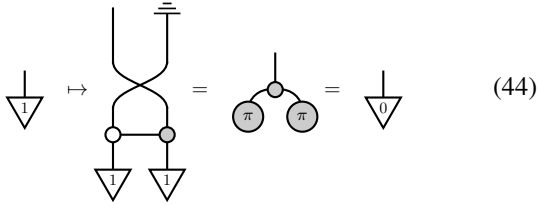
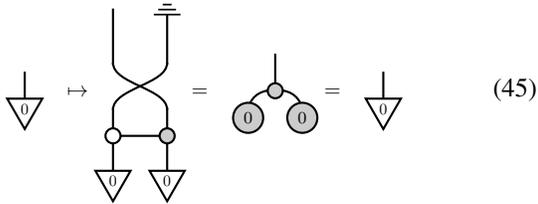
\begin{equation}
\scalebox{0.8}{$
	\begin{tikzpicture}
	\begin{pgfonlayer}{nodelayer}
		\node [style=dpoint] (0) at (0, -0.5) {$1$};
		\node [style=none] (1) at (0, 0.5) {};
	\end{pgfonlayer}
	\begin{pgfonlayer}{edgelayer}
		\draw [style=boldedge] (1.center) to (0);
	\end{pgfonlayer}
\end{tikzpicture}
	\hspace{3mm}
	\mapsto
	\begin{tikzpicture}
	\begin{pgfonlayer}{nodelayer}
		\node [style=none] (0) at (-0.75, 2.25) {};
		\node [style=white ddot] (1) at (-0.75, -1.25) {};
		\node [style=dpoint] (2) at (0.75, -2.5) {$1$};
		\node [style=none] (3) at (-0.75, 0.75) {};
		\node [style=none] (4) at (0.75, -0.5) {};
		\node [style=none] (5) at (0.75, 0.75) {};
		\node [style=none] (6) at (-0.75, -0.5) {};
		\node [style=gray ddot] (7) at (0.75, -1.25) {};
		\node [style=dpoint] (8) at (-0.75, -2.5) {$1$};
		\node [style=upground] (9) at (0.75, 2.25) {};
	\end{pgfonlayer}
	\begin{pgfonlayer}{edgelayer}
		\draw [style=boldedge, in=-90, out=90, looseness=0.75] (6.center) to (5.center);
		\draw [style=boldedge, in=90, out=-90, looseness=0.75] (3.center) to (4.center);
		\draw [style=boldedge] (5.center) to (9);
		\draw [style=boldedge] (4.center) to (7);
		\draw [style=boldedge] (7) to (2);
		\draw [style=boldedge] (3.center) to (0.center);
		\draw [style=boldedge] (6.center) to (1);
		\draw [style=boldedge] (1) to (7);
		\draw [style=boldedge] (1) to (8);
	\end{pgfonlayer}
\end{tikzpicture}
	=
	\hspace{3mm}
	\begin{tikzpicture}
	\begin{pgfonlayer}{nodelayer}
		\node [style=gray ddot] (0) at (0, 0) {};
		\node [style=none] (1) at (0, 1) {};
		\node [style=gray phase ddot] (2) at (-0.75, -0.75) {$\pi$};
		\node [style=gray phase ddot] (3) at (0.75, -0.75) {$\pi$};
	\end{pgfonlayer}
	\begin{pgfonlayer}{edgelayer}
		\draw [style=boldedge] (0) to (1.center);
		\draw [style=boldedge, bend left=45, looseness=1.00] (2) to (0);
		\draw [style=boldedge, bend left=45, looseness=1.00] (0) to (3);
	\end{pgfonlayer}
\end{tikzpicture}
	=
	\hspace{3mm}
	\begin{tikzpicture}
	\begin{pgfonlayer}{nodelayer}
		\node [style=dpoint] (0) at (0, -0.5) {$0$};
		\node [style=none] (1) at (0, 0.5) {};
	\end{pgfonlayer}
	\begin{pgfonlayer}{edgelayer}
		\draw [style=boldedge] (1.center) to (0);
	\end{pgfonlayer}
\end{tikzpicture}
$}
\end{equation}
\begin{equation}
\scalebox{0.8}{$
	\begin{tikzpicture}
	\begin{pgfonlayer}{nodelayer}
		\node [style=dpoint] (0) at (0, -0.5) {$0$};
		\node [style=none] (1) at (0, 0.5) {};
	\end{pgfonlayer}
	\begin{pgfonlayer}{edgelayer}
		\draw [style=boldedge] (1.center) to (0);
	\end{pgfonlayer}
\end{tikzpicture}
	\hspace{3mm}
	\mapsto
	\begin{tikzpicture}
	\begin{pgfonlayer}{nodelayer}
		\node [style=gray ddot] (0) at (1, -1.25) {};
		\node [style=dpoint] (1) at (1, -2.5) {$0$};
		\node [style=none] (2) at (-0.5, 2.25) {};
		\node [style=none] (3) at (1, -0.5) {};
		\node [style=upground] (4) at (1, 2.25) {};
		\node [style=none] (5) at (-0.5, 0.75) {};
		\node [style=dpoint] (6) at (-0.5, -2.5) {$0$};
		\node [style=none] (7) at (1, 0.75) {};
		\node [style=white ddot] (8) at (-0.5, -1.25) {};
		\node [style=none] (9) at (-0.5, -0.5) {};
	\end{pgfonlayer}
	\begin{pgfonlayer}{edgelayer}
		\draw [style=boldedge, in=-90, out=90, looseness=0.75] (9.center) to (7.center);
		\draw [style=boldedge, in=90, out=-90, looseness=0.75] (5.center) to (3.center);
		\draw [style=boldedge] (7.center) to (4);
		\draw [style=boldedge] (3.center) to (0);
		\draw [style=boldedge] (0) to (1);
		\draw [style=boldedge] (5.center) to (2.center);
		\draw [style=boldedge] (9.center) to (8);
		\draw [style=boldedge] (8) to (0);
		\draw [style=boldedge] (8) to (6);
	\end{pgfonlayer}
\end{tikzpicture}
	=
	\hspace{3mm}
	\begin{tikzpicture}
	\begin{pgfonlayer}{nodelayer}
		\node [style=gray phase ddot] (0) at (0.75, -0.75) {$0$};
		\node [style=gray ddot] (1) at (0, 0) {};
		\node [style=gray phase ddot] (2) at (-0.75, -0.75) {$0$};
		\node [style=none] (3) at (0, 1) {};
	\end{pgfonlayer}
	\begin{pgfonlayer}{edgelayer}
		\draw [style=boldedge] (1) to (3.center);
		\draw [style=boldedge, bend left=45, looseness=1.00] (2) to (1);
		\draw [style=boldedge, bend left=45, looseness=1.00] (1) to (0);
	\end{pgfonlayer}
\end{tikzpicture}
	=
	\hspace{3mm}
	\begin{tikzpicture}
	\begin{pgfonlayer}{nodelayer}
		\node [style=dpoint] (0) at (0, -0.5) {$0$};
		\node [style=none] (1) at (0, 0.5) {};
	\end{pgfonlayer}
	\begin{pgfonlayer}{edgelayer}
		\draw [style=boldedge] (1.center) to (0);
	\end{pgfonlayer}
\end{tikzpicture}
$}
\end{equation}
However, things get somewhat weird if we add some amount $\epsilon$ of noise to the input state $\rho = \ket{0}\bra{0}$:
\begin{equation}
	\rho' = \dfrac{\epsilon}{2} \mathbb{1} + (1-\epsilon)\ket{0}\bra{0} = (1- \dfrac{\epsilon}{2})\ket{0}\bra{0} + \dfrac{\epsilon}{2} \ket{1}\bra{1}
\end{equation}
The fixed-point equation for this new input state becomes:
\begin{equation}
	\tikzfigscale{0.8}{tauprime} = (1-\dfrac{\epsilon}{2}) \tikzfigscale{0.8}{nonlinearity_2}+ \dfrac{\epsilon}{2}\tikzfigscale{0.8}{nonlinearity_20} = (1-\dfrac{\epsilon}{2}) \tikzfigscale{0.8}{state1}+\dfrac{\epsilon}{2}\tikzfigscale{0.8}{state0}
\end{equation}
Plugging the new fixed-point state $\tau'$ in, we see that the D-CTC map sends the perturbed state $\rho'$ to:
\begin{equation}
	\epsilon(1-\dfrac{\epsilon}{2})\tikzfigscale{0.8}{state1} + (1-\epsilon + \dfrac{\epsilon^2}{2})\tikzfigscale{0.8}{state0}
\end{equation}
From this we can conclude that maps in the D-CTC model are not necessarily linear.

\subsection{Discontinuity}

\noindent In order to show that quantum maps involving D-CTCs can be discontinuous as well as non-linear, consider the following:
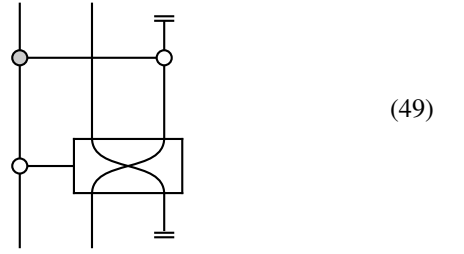
\begin{equation}
	\tikzfigscale{0.8}{circuit}
\end{equation}
Where the first gate denotes a \emph{controlled SWAP} unitary gate. For input $\ket{0}\bra{0} \otimes \rho$, the fixed-point equation prescribed by the D-CTC model is as follows:
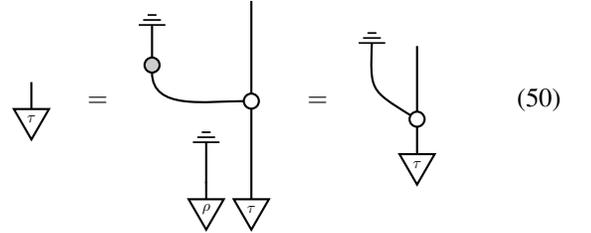
\begin{equation}
	\tikzfigscale{0.8}{tau}
	\hspace{3mm}
	=
	\hspace{3mm}
	\tikzfigscale{0.8}{discon2}
	\hspace{3mm}
	=
	\hspace{3mm}
	\tikzfigscale{0.8}{discon3}
\end{equation}
meaning that every state invariant under decoherence in the Z basis is a fixed point. This means that the fixed point of maximal entropy is simply the maximally mixed state, and we get the following evolution:
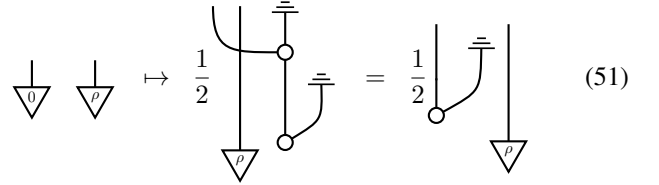
\begin{equation}
	\tikzfigscale{0.8}{state0} \ \ \tikzfigscale{0.8}{staterho}
	\hspace{2mm}
	\mapsto
	\hspace{2mm}
	\dfrac{1}{2} \tikzfigscale{0.8}{evolution_discontinuity}
	\hspace{2mm}
	=
	\hspace{2mm}
	\dfrac{1}{2} \tikzfigscale{0.8}{evolution_discontinuity2}
\end{equation}
For input $\ket{1}\bra{1} \otimes \rho$, the fixed-point equation prescribed by the D-CTC model is as follows:
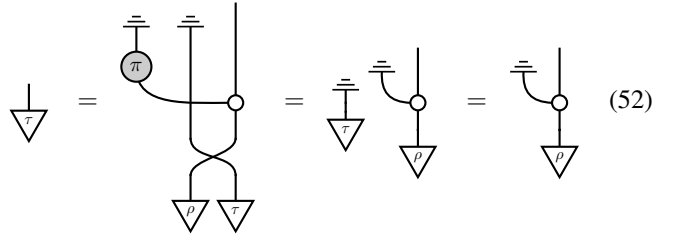
\begin{equation}
	\tikzfigscale{0.8}{tau}
	\hspace{2mm}
	=
	\hspace{2mm}
	\tikzfigscale{0.8}{discon2new}
	\hspace{2mm}
	=
	\hspace{2mm}
	\tikzfigscale{0.8}{discon3new}
	\hspace{2mm}
	=
	\hspace{2mm}
	\tikzfigscale{0.8}{discon4new}
\end{equation}
There is a unique fixed point in this case, namely the decoherence of $\rho$ in the Z basis, hence we get the following evolution:
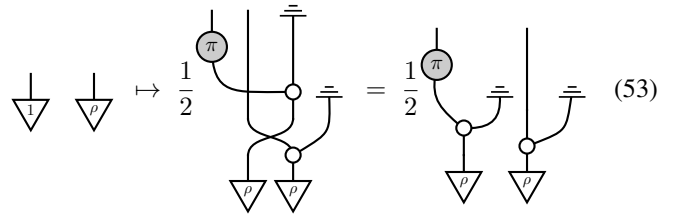
\begin{equation}
	\tikzfigscale{0.8}{state1} \ \ \tikzfigscale{0.8}{staterho}
	\hspace{1mm}
	\mapsto
	\hspace{1mm}
	\dfrac{1}{2} \tikzfigscale{0.8}{evolution_discontinuity_new}
	\hspace{1mm}
	=
	\hspace{1mm}
	\dfrac{1}{2} \tikzfigscale{0.8}{evolution_discontinuity2_new}
\end{equation}
We can now consider the evolution of the mixture $[(1-\epsilon) \ket{0} \bra{0} + \epsilon \ket{1} \bra{1}]  \otimes \rho$ for all $\epsilon \in (0,1)$, getting the following fixed-point equation:
\begin{equation}
	\tikzfigscale{0.8}{tau}
	\hspace{2mm}
	=
	\hspace{2mm}
	(1-\epsilon)\tikzfigscale{0.8}{discon3} + \epsilon \tikzfigscale{0.8}{discon3new}
\end{equation}
After a few calculations, this again yields the decoherence of $\rho$ in the Z basis as the unique solution:
\begin{equation}
	\tikzfigscale{0.8}{tau}=\tikzfigscale{0.8}{discon3newnew}
\end{equation}
This yields the following evolution for $\epsilon \in (0,1)$:
\begin{equation}
	(1- \epsilon) \ \tikzfigscale{0.8}{finalevol}+ \epsilon \ \  \tikzfigscale{0.8}{evolution_discontinuity2_new}
\end{equation}
If we let $\epsilon$ go to $0$, we see that the limiting value of the evolution is:
\begin{equation}
	\tikzfigscale{0.8}{finalevol} = \tikzfigscale{0.8}{finalevol2}
\end{equation}
This is different from the value of the evolution on initial state $\ket{0} \bra{0} \otimes \rho$, showing a discontinuity at $\epsilon = 0$.

\section{Cloning in the D-CTC model}
\label{appendix:cloning}

\noindent In \cite{brun2013}, the authors argue that the fact itself of sending a state inside the following CTC could be seen as a cloner,
\begin{equation}
	\tikzfigscale{0.8}{circuitcloning}
\end{equation}
but also observe that the $N$ copies inside the CTC are not actually available after the system leaves the ``wormhole''. Indeed, the fixed point for the map above is easily computed to be $\rho^{\otimes N}$, so that $N$ clones of the input state $\rho$ are truly created inside the CTC. However, the copies are never accessible from the outside, and so this is not actually an implementation of a cloner for the purposes of the CR region.

One could try to extract the copies by performing an interaction with the CTC involving a series of CNOTs:
\begin{equation}
	\tikzfigscale{0.8}{circuitcloning2}
\end{equation}
However, this results only in the creation of $N$ copies of the decoherence of $\rho$ in the Z basis (as can be seen in the following example with $N=3$):
\begin{align}
	&\tikzfigscale{0.8}{circuitcloning_stefano} \\
	= &\tikzfigscale{0.8}{circuitcloning_stefano3} \nonumber
\end{align}
This is already significant in itself: it allows cloning of classical probabilistic states (which are encoded in the Z basis, cloned and then measured again in the Z basis), something which is impossible in classical theory alone. It is not, however, a full quantum cloner.

To obtain a quantum cloner---or, more precisely, an approximate quantum cloner reaching full cloning fidelity in the limit $N \rightarrow \infty$---one can replace the decoherence in the Z basis:
\[
	\rho \mapsto \sum_{i = 0}^{d-1} \operatorname{Tr}\left[\ket{i}\bra{i} \rho\right] \ket{i}\bra{i}
\]
with a symmetric informationally complete positive-operator--valued measurement (SIC-POVM):
\[
	\rho \mapsto \sum_{x = 0}^{d^2-1} \operatorname{Tr}\left[M_x \rho\right] \ket{x} \bra{x}
\]
In the limit $N \rightarrow \infty$, measuring $N$ copies of the resulting state in the $\left(\ket{x}\bra{x}\right)_{x=0}^{d^2-1}$ basis results in full tomography of the state $\rho$, which can subsequently be cloned at will \cite{brun2013}.

\newcommand{\FramedCausalGraphsCategory}{\textbf{CausGraphs}}
\newcommand{\CRFramedCausalGraphsCategory}{\textbf{CausGraphs}^{\text{CR}}}
\newcommand{\CVLocalFramedCausalGraphsCategory}{\textbf{CausGraphs}^{\text{CV-Loc}}}

\section{The SMC of framed causal graphs}
\label{appendix:framedCausalGraphsSMC}

\noindent The \emph{category of framed causal graphs} $\FramedCausalGraphsCategory$ has the natural numbers as its objects. The morphisms $n \rightarrow m$ in $\FramedCausalGraphsCategory$ are the framed causal graphs $G$ with $\#\inputs{G} = n$ and $\#\outputs{G} = m$. Using the framing, we can canonically identify $\inputs{G}$ with the total order $\{0,...,n-1\}$, and $\inputs{G}$ with the total order $\{0,...,m-1\}$.

Composition $H \circ G$ of morphisms $G:n\rightarrow m$ and $H: m \rightarrow r$ in the category is given by gluing the intermediate ``open ends'' $\outputs{G}$ and $\inputs{H}$. Specifically, the graph $H\circ G$ has the following nodes:
\begin{itemize}
	\item $\nodes{H \circ G} := \nodes{G} \sqcup \nodes{H}$;
	\item $\inputs{H \circ G} := \inputs{G}$;
	\item $\outputs{H \circ G} := \outputs{H}$.
\end{itemize}
An edge $x\rightarrow y$ is in $H\circ G$ if and only if one of the following conditions holds:
\begin{enumerate}
	\item[(i)] $x\rightarrow y$ in $G$ and  $y \notin \outputs{G}$;
	\item[(ii)] $x\rightarrow y$ in $H$  and  $x \notin \inputs{H}$;
	\item[(iii)] there exists a $b \in \{0,...,m-1\}$ s.t. $x\rightarrow b$ in $G$ and $b\rightarrow y$ in $H$, where we have identified both $\outputs{G}$ and $\inputs{H}$ with the total order $\{0,...,m-1\}$.
\end{enumerate}
The identity $\id{A}: A \rightarrow A$ on a total order $A$ is given by the digraph with $A\times\{0\} \sqcup A\times\{1\}$ as set of nodes and $((a,0),(a,1))$ for all $a \in A$ as edges.
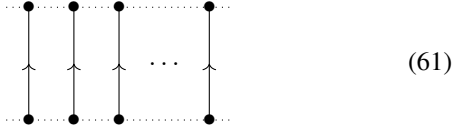
\begin{equation}
	\begin{tikzpicture}
	\begin{pgfonlayer}{nodelayer}
		\node [style=none] (0) at (-2, 1.25) {$\bullet$};
		\node [style=none] (1) at (2, 1.25) {$\bullet$};
		\node [style=none] (2) at (2, -1.25) {$\bullet$};
		\node [style=none] (3) at (-2, -1.25) {$\bullet$};
		\node [style=none] (4) at (-1, 1.25) {$\bullet$};
		\node [style=none] (5) at (-1, -1.25) {$\bullet$};
		\node [style=none] (6) at (1, 0) {$\ldots$};
		\node [style=none] (7) at (-2.5, 1.25) {};
		\node [style=none] (8) at (2.5, 1.25) {};
		\node [style=none] (9) at (-2.5, -1.25) {};
		\node [style=none] (10) at (2.5, -1.25) {};
		\node [style=none] (11) at (0, 1.25) {$\bullet$};
		\node [style=none] (12) at (0, -1.25) {$\bullet$};
	\end{pgfonlayer}
	\begin{pgfonlayer}{edgelayer}
		\draw [style=->-] (3.center) to (0.center);
		\draw [style=->-] (5.center) to (4.center);
		\draw [style=->-] (2.center) to (1.center);
		\draw [dotted] (7.center) to (8.center);
		\draw [dotted] (9.center) to (10.center);
		\draw [style=->-] (12.center) to (11.center);
	\end{pgfonlayer}
\end{tikzpicture}
\end{equation}
The category $\FramedCausalGraphsCategory$ can be endowed with the following symmetric monoidal structure $(\FramedCausalGraphsCategory,\oplus,\emptyset,\sigma)$:
\begin{itemize}
	\item on objects, $A \oplus B$ is the sum total order $A + B$, where all elements of $A$ are taken to come before all elements of $B$;
	\item on morphisms, $G \oplus H$ is the disjoint union $G \sqcup H$ of digraphs $G$ and $H$, with $\inputs{G \oplus H} := \inputs{G} \sqcup \inputs{H}$ and $\outputs{G \oplus H} := \outputs{G} \sqcup \outputs{H}$;
	\item the tensor unit $\emptyset$ is the empty digraph, with $\inputs{\emptyset} = \emptyset = \outputs{\emptyset}$;
	\item the symmetry isomorphism $\sigma_{A,B}: A \oplus B \rightarrow B \oplus A$ is the digraph with $(A+B)\times\{0\} \sqcup (B+A)\times\{1\}$ as its set of nodes, and edges $((a,0),(a,1))$ and $((b,0),(b,1))$ for all $a \in A$ and all $b \in B$.
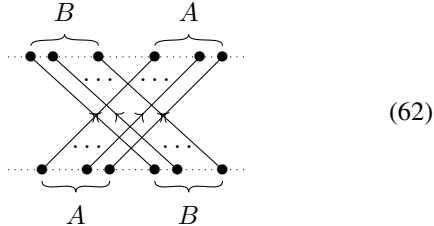
	\begin{equation}
		\begin{tikzpicture}
	\begin{pgfonlayer}{nodelayer}
		\node [style=none] (0) at (2.5, -1.25) {};
		\node [style=none] (1) at (1.5, 1.25) {$\bullet$};
		\node [style=none] (2) at (-2, -1.25) {$\bullet$};
		\node [style=none] (3) at (2.5, 1.25) {};
		\node [style=none] (4) at (-1.75, 1.25) {$\bullet$};
		\node [style=none] (5) at (-0.75, 1.25) {$\bullet$};
		\node [style=none] (6) at (0.5, 1.25) {$\bullet$};
		\node [style=none] (7) at (-2.75, 1.25) {};
		\node [style=none] (8) at (1, -1.25) {$\bullet$};
		\node [style=none] (9) at (2, -1.25) {$\bullet$};
		\node [style=none] (10) at (-1, -1.25) {$\bullet$};
		\node [style=none] (11) at (-2.75, -1.25) {};
		\node [style=none] (12) at (-2.25, 1.25) {$\bullet$};
		\node [style=none] (13) at (0.5, -1.25) {$\bullet$};
		\node [style=none] (14) at (-0.5, -1.25) {$\bullet$};
		\node [style=none] (15) at (2, 1.25) {$\bullet$};
		\node [style=none] (16) at (-1, -0.75) {$\ldots$};
		\node [style=none] (17) at (0.5, 0.75) {$\ldots$};
		\node [style=none] (18) at (1, -0.75) {$\ldots$};
		\node [style=none] (19) at (-0.75, 0.75) {$\ldots$};
		\node [style=none] (20) at (-2.25, 1.5) {};
		\node [style=none] (21) at (-0.75, 1.5) {};
		\node [style=none] (22) at (-1.5, 1.75) {};
		\node [style=none] (23) at (0.5, 1.5) {};
		\node [style=none] (24) at (2, 1.5) {};
		\node [style=none] (25) at (1.25, 1.75) {};
		\node [style=none] (26) at (-2, -1.5) {};
		\node [style=none] (27) at (-0.5, -1.5) {};
		\node [style=none] (28) at (-1.25, -1.75) {};
		\node [style=none] (29) at (0.5, -1.5) {};
		\node [style=none] (30) at (1.25, -1.75) {};
		\node [style=none] (31) at (2, -1.5) {};
		\node [style=none] (32) at (-1.25, -2.25) {$A$};
		\node [style=none] (33) at (1.25, 2.25) {$A$};
		\node [style=none] (34) at (1.25, -2.25) {$B$};
		\node [style=none] (35) at (-1.5, 2.25) {$B$};
	\end{pgfonlayer}
	\begin{pgfonlayer}{edgelayer}
		\draw [style=->-] (2.center) to (6.center);
		\draw [style=->-] (10.center) to (1.center);
		\draw [style=->-] (9.center) to (5.center);
		\draw [dotted] (7.center) to (3.center);
		\draw [dotted] (11.center) to (0.center);
		\draw [style=->-] (8.center) to (4.center);
		\draw [style=->-] (13.center) to (12.center);
		\draw [style=->-] (14.center) to (15.center);
		\draw [in=-90, out=75, looseness=0.75] (20.center) to (22.center);
		\draw [in=90, out=-90, looseness=0.75] (22.center) to (21.center);
		\draw [in=-90, out=75, looseness=0.75] (23.center) to (25.center);
		\draw [in=90, out=-90, looseness=0.75] (25.center) to (24.center);
		\draw [in=90, out=-90, looseness=1.00] (26.center) to (28.center);
		\draw [in=-90, out=90, looseness=1.00] (28.center) to (27.center);
		\draw [in=90, out=-90, looseness=1.00] (29.center) to (30.center);
		\draw [in=-90, out=90, looseness=1.00] (30.center) to (31.center);
	\end{pgfonlayer}
\end{tikzpicture}
	\end{equation}
\end{itemize}

Because of the way composition is defined in $\FramedCausalGraphsCategory$, framed causal graphs cannot be used to describe scenarios in which inputs/outputs live in a chronology-violating region: no new cycles can ever be created by sequential or parallel composition. From a physical perspective, this means that framed causal graphs can be used to describe regions of spacetime containing CTCs, but only with boundaries constrained to live in the chronology-respecting sector.

\section{Proofs}
\label{appendix:proofs}

\newcounter{theorem:save}
\setcounter{theorem:save}{\value{theorem_c}}

\setcounter{theorem_c}{\value{theorem:DMixSMC}}
\begin{theorem}
	The category \DMix is a symmetric monoidal category equipped with the following tensor product (and the swap CPTP maps from quantum theory):
	\begin{equation}
		\tikzfigscale{0.8}{monoidal_product}
		\hspace{4mm}
		:=
		\hspace{3mm}
		\tikzfigscale{0.8}{monoidal_product2}
	\end{equation}
\end{theorem}
\begin{proof}
	Proving that the monoidal product above is well-defined essentially reduces to showing that the following equation holds for arbitrary $\Psi$, $f$ and $g$:
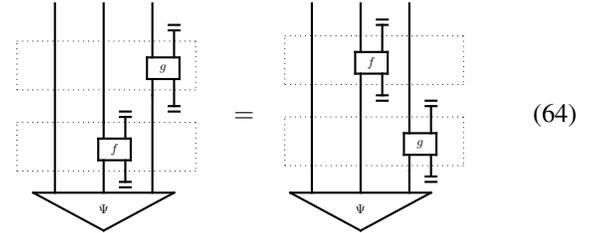
	\begin{equation}
		\tikzfigscale{0.8}{equation_monoidal_product1}
		\hspace{3mm} = \hspace{3mm}
		\tikzfigscale{0.8}{equation_monoidal_product_new}
	\end{equation}
	By definition of the D-CTC model, the LHS can be rewritten as follows:
	\begin{equation}
		\tikzfigscale{0.8}{equation_monoidal_product1} \hspace{3mm} = \hspace{3mm} \tikzfigscale{0.8}{equation_monoidal_product4}
	\end{equation}
	where $\tau$ and $\sigma$ satisfy the following fixed-point equations:
	\begin{equation}
		\tikzfigscale{0.6}{equation_monoidal_product_fp} \hspace{3mm} = \hspace{3mm} \tikzfigscale{0.6}{tau}
	\end{equation}
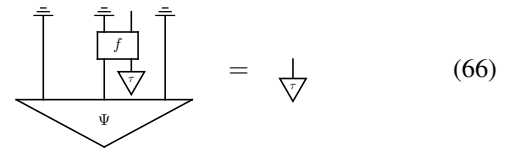
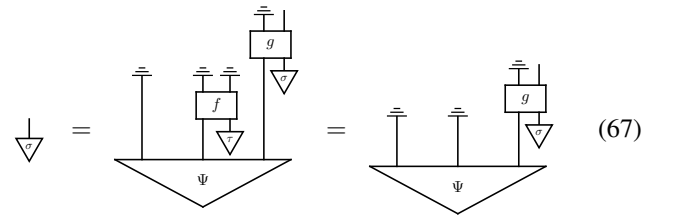
	\begin{equation}
		\tikzfigscale{0.6}{sigma} \hspace{2mm} = \hspace{2mm} \tikzfigscale{0.6}{equation_monoidal_product_fpsigma} \hspace{2mm} = \hspace{2mm} \tikzfigscale{0.6}{equation_monoidal_product_fpsigma2}
	\end{equation}
	Similarly, the RHS can be rewritten as follows:
	\begin{equation}
		\tikzfigscale{0.8}{equation_monoidal_product_new} \hspace{3mm} = \hspace{3mm} \tikzfigscale{0.8}{equation_monoidal_product_newfp}
	\end{equation}
	where $\tau'$ and $\sigma'$ satisfy the same fixed-point equations as $\tau$ and $\sigma$. For example, the equation for $\sigma'$ is:
	\begin{equation}
		 \tikzfigscale{0.6}{sigmaprime} \hspace{3mm} = \hspace{3mm} \tikzfigscale{0.6}{equation_monoidal_product_fpsigma2prime}
	\end{equation}
	We are then left with simple sliding of CPTP maps:
	\begin{equation}
		\tikzfigscale{0.8}{equation_product_new} \hspace{2mm} = \hspace{2mm} \tikzfigscale{0.8}{equation_product_new3} \hspace{2mm} = \hspace{2mm}  \tikzfigscale{0.8}{equation_product_new2}
	\end{equation}
\end{proof}

\setcounter{theorem_c}{\value{lemma:MixDMix}}
\begin{lemma}
	The category \Mix of finite-dimensional Hilbert spaces and CPTP maps is faithfully embedded in \DMix by the following monoidal functor:
	\begin{equation}
		\tikzfigscale{0.8}{embedding1new}
		\hspace{3mm}
		\mapsto
		\hspace{3mm}
		\tikzfigscale{0.8}{embedding2new}
	\end{equation}
\end{lemma}
\begin{proof}
	We begin by checking that this mapping, which we shall call $F$, is a functor, i.e. that:
	\[
	F(f) \circ_{\DMix} F(g) = F(f \circ g)
	\hspace{10mm}
	F(\id{}) = \id{}
	\]
	Indeed, the LHS of the composition equation is just the composition of the two elementary boxes:
	\[
	F(f) \circ_{\DMix} F(g) = \tikzfigscale{0.8}{graphical_composition_Ff_Fg}
	\]
	while the RHS is written as follows:
	\[
	F(f \circ g) = \tikzfigscale{0.8}{graphical_composition_Ffg}
	\]
	To show that the two morphisms are equal, we need to show that they are both in the same equivalence class for the relation $\sim$, i.e that for an arbitrary auxiliary system and bipartite state $\rho$ we have:
	\[
	\tikzfigscale{0.8}{graphical_composition_Ff_Fg2}
	\hspace{3mm}
	=
	\hspace{3mm}
	\tikzfigscale{0.8}{graphical_composition_Ffg2}
	\]
	Because the only normalized state on the tensor unit $\mathbb{C}$ is the scalar $1$, the fixed-point equations are all trivial, and the output states for both maps are equal to $\big(\id{} \otimes (f \circ g)\big)(\Phi)$. Similarly one can show that the identity of $\Mix$ is sent by $F$ to the identity of $\DMix$, and that $F(\id{} \otimes f) = \id{} \otimes F(f)$. The latter in turn implies that the functor is monoidal.

	To show that the functor is faithful, suppose $f \neq g$. Then, there is a state $\rho$ such that $f(\rho) \neq g(\rho)$ in $\Mix$. This immediately yields the following inequality, proving that $F(f) \not{\hspace{-2.5mm}\sim} F(g)$:
	\begin{equation}
		F\big(g(\rho)\big) = \tikzfigscale{0.8}{grho} \ \  \neq \ \ \tikzfigscale{0.8}{frho} = F\big(f(\rho)\big)
	\end{equation}

	The triviality of fixed-point equations involving the tensor unit $\mathbb{C}$ means, in particular, that a CTC containing the identity on the tensor unit yields the scalar 1. This justifies the following notational convention:
	\begin{equation}
		\tikzfigscale{0.8}{convention} \ \ = \ \ \tikzfigscale{0.8}{convention2}
	\end{equation}
\end{proof}

\setcounter{theorem_c}{\value{theorem:DMixTerminal}}
\begin{lemma}
	The category \DMix is terminal, i.e. on any given system $\mathcal{H}$ the only effect is the discarding map inherited from ordinary quantum theory:
	\begin{equation}
		\tikzfigscale{0.8}{discarding}
	\end{equation}
	In particular, the D-CTC model respects both no-signaling and causality (no-signaling from the future).
\end{lemma}
\begin{proof}
	By definition of $\DMix$, it suffices to show that elementary boxes satisfy the following equivalence as CPTP maps:
	\[
		\tikzfigscale{0.8}{terminality0new} \sim \tikzfigscale{0.8}{discarding}
	\]
	This is immediate to check against an arbitrary bipartite state $\rho$, using the fact that $f$ is a CPTP map itself, and that the fixed-point state $\tau$ is normalized by definition:
	\[
		\tikzfigscale{0.8}{terminality1} = \tikzfigscale{0.8}{terminality2} =\tikzfigscale{0.8}{terminality3}=\tikzfigscale{0.8}{terminality4}
	\]
\end{proof}

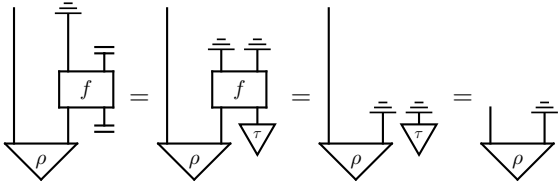
\setcounter{theorem_c}{\value{lemma:causalSetsNontransAcyclicDiraphs}}
\begin{lemma}
	A causal set is the same as a non-transitive acyclic digraph (directed graph).
\end{lemma}
\begin{proof}
	Starting from a causal set $(C, \preceq)$, a non-transitive acyclic digraph is constructed with the elements of $C$ as its nodes, and edges $v_0 \rightarrow v_1$ for all $v_0,v_1 \in C$ such that (i) $v_0 \prec v_1$ and (ii) there is no $z \in C$ such that $v_0 \prec z \prec v_1$. The fact that $\preceq$ is a partial order implies that the resulting graph is acyclic, and the fact that it is locally finite implies that $\prec$ can be recovered as the transitive closure of $\rightarrow$.
	Starting from a non-transitive acyclic digraph, a causal set is constructed with elements $C$ equal to the nodes of the graph, and $\preceq$ the reflexive and transitive closure of directed edge relation $\rightarrow$ in the digraph. The fact that the graph is acyclic implies that $\preceq$ is a partial order, while the fact that the graph is non-transitive implies that $\preceq$ is locally finite.
	The two maps are evidently inverse of each other, completing our proof.
\end{proof}

\setcounter{theorem_c}{\value{theorem:DMixProperties}}
\begin{theorem}
	D-CTCs and the category $\DMix$ define a time-travel super-operator for quantum theory, i.e. for the symmetric monoidal category $\Mix$ of finite-dimensional Hilbert spaces and CPTP maps. The same is true of P-CTCs and the category $\Mix_{sym}$.
\end{theorem}
\begin{proof}
	It is straightforward to check that traced monoidal categories \cite{selinger2011} satisfy the properties required for P-CTCs and $\Mix_{sym}$ to define a time-travel super-operator for quantum theory. We have now to show that the super-operator defined by the D-CTC model satisfies the four properties required of time-travel super-operators. This is shown by Lemmas \ref{lemma:deutsch_naturality}, \ref{lemma:deutsch_strength}, \ref{lemma:deutsch_sliding} and \ref{lemma:deutsch_vanishing} below.
\end{proof}

\setcounter{theorem_c}{\value{theorem:save}}

\begin{lemma}
\label{lemma:deutsch_naturality}
	The super-operator defined in the D-CTC model satisfies the property of naturality in the CR region:
	
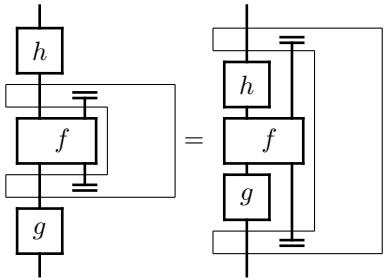
\[
		\begin{tikzpicture}
	\begin{pgfonlayer}{nodelayer}
		\node [style=none] (0) at (-0.5, -0.5) {};
		\node [style=none] (1) at (0.5, 0.5) {};
		\node [style=none] (2) at (-1, -0.5) {};
		\node [style=none] (3) at (-0.5, -1.25) {};
		\node [style=none] (4) at (0.75, -1.1) {};
		\node [style=none] (5) at (0, -0.5) {};
		\node [style=none] (6) at (0.5, 0.95) {};
		\node [style=none] (7) at (0.25, 0.95) {};
		\node [style=none] (8) at (0.75, 1.075) {};
		\node [style=none] (9) at (0, 0) {$f$};
		\node [style=none] (10) at (-1, 0.5) {};
		\node [style=none] (11) at (0.5, -0.975) {};
		\node [style=none] (12) at (0, 0.5) {};
		\node [style=none] (13) at (-0.5, -1.5) {};
		\node [style=none] (14) at (-0.5, 0.5) {};
		\node [style=none] (15) at (0.5, -0.5) {};
		\node [style=none] (16) at (-0.5, 1.5) {};
		\node [style=none] (17) at (-0.5, -0.5) {};
		\node [style=none] (18) at (0.75, 0.95) {};
		\node [style=none] (19) at (0.25, -0.975) {};
		\node [style=none] (20) at (0.75, -0.975) {};
		\node [style=none] (21) at (0.75, -0.5) {};
		\node [style=none] (22) at (0.75, 0.5) {};
		\node [style=none] (23) at (0.25, 1.075) {};
		\node [style=none] (24) at (-0.5, -1.25) {};
		\node [style=none] (25) at (0.25, -1.1) {};
		\node [style=none] (26) at (-1.25, 0.75) {};
		\node [style=none] (27) at (-1.25, -1.25) {};
		\node [style=none] (28) at (1, 0.75) {};
		\node [style=none] (29) at (2.5, -1.25) {};
		\node [style=none] (30) at (1, -0.75) {};
		\node [style=none] (31) at (-1.25, -0.75) {};
		\node [style=none] (32) at (2.5, 1.25) {};
		\node [style=none] (33) at (-1.25, 1.25) {};
		\node [style=none] (34) at (-0.5, -1.5) {};
		\node [style=none] (35) at (-0.75, -2.5) {};
		\node [style=none] (36) at (0, -2.5) {};
		\node [style=none] (37) at (-0.5, -2) {$g$};
		\node [style=none] (38) at (-1, -1.5) {};
		\node [style=none] (39) at (-0.25, -1.5) {};
		\node [style=none] (40) at (-0.75, -1.5) {};
		\node [style=none] (41) at (-1, -2.5) {};
		\node [style=none] (42) at (-0.75, -2.5) {};
		\node [style=none] (43) at (-0.25, -2.5) {};
		\node [style=none] (44) at (0, -1.5) {};
		\node [style=none] (45) at (-0.5, -2.5) {};
		\node [style=none] (46) at (-0.5, -2.5) {};
		\node [style=none] (47) at (-0.5, -3) {};
		\node [style=none] (48) at (-0.5, 2.5) {};
		\node [style=none] (49) at (-0.75, 2.5) {};
		\node [style=none] (50) at (-0.5, 2) {$h$};
		\node [style=none] (51) at (-1, 1.5) {};
		\node [style=none] (52) at (-1, 2.5) {};
		\node [style=none] (53) at (0, 2.5) {};
		\node [style=none] (54) at (-0.25, 2.5) {};
		\node [style=none] (55) at (-0.5, 3) {};
		\node [style=none] (56) at (-0.5, 2.5) {};
		\node [style=none] (57) at (0, 1.5) {};
		\node [style=none] (58) at (-0.5, 1.5) {};
		\node [style=none] (59) at (-0.75, 1.5) {};
		\node [style=none] (60) at (-0.5, 1.5) {};
		\node [style=none] (61) at (-0.5, 1.5) {};
		\node [style=none] (62) at (-0.5, 3) {};
		\node [style=none] (63) at (-0.75, 1.5) {};
		\node [style=none] (64) at (-0.25, 1.5) {};
	\end{pgfonlayer}
	\begin{pgfonlayer}{edgelayer}
		\draw [style=boldedge] (13.center) to (24.center);
		\draw [style=boldedge] (10.center) to (2.center);
		\draw [style=boldedge] (21.center) to (2.center);
		\draw [style=boldedge] (21.center) to (22.center);
		\draw [style=boldedge] (22.center) to (10.center);
		\draw [style=boldedge] (14.center) to (16.center);
		\draw [style=boldedge] (1.center) to (6.center);
		\draw [style=boldedge] (11.center) to (15.center);
		\draw [style=boldedge] (23.center) to (8.center);
		\draw [style=boldedge] (18.center) to (7.center);
		\draw [style=boldedge] (20.center) to (19.center);
		\draw [style=boldedge] (25.center) to (4.center);
		\draw [style=boldedge] (3.center) to (0.center);
		\draw (29.center) to (32.center);
		\draw (32.center) to (33.center);
		\draw (33.center) to (26.center);
		\draw (26.center) to (28.center);
		\draw (28.center) to (30.center);
		\draw (30.center) to (31.center);
		\draw (31.center) to (27.center);
		\draw (27.center) to (29.center);
		\draw [style=boldedge] (38.center) to (41.center);
		\draw [style=boldedge] (36.center) to (41.center);
		\draw [style=boldedge] (36.center) to (44.center);
		\draw [style=boldedge] (44.center) to (38.center);
		\draw [style=boldedge] (46.center) to (47.center);
		\draw [style=boldedge] (48.center) to (62.center);
		\draw [style=boldedge] (52.center) to (51.center);
		\draw [style=boldedge] (57.center) to (51.center);
		\draw [style=boldedge] (57.center) to (53.center);
		\draw [style=boldedge] (53.center) to (52.center);
		\draw [style=boldedge] (61.center) to (60.center);
	\end{pgfonlayer}
\end{tikzpicture} = \begin{tikzpicture}
	\begin{pgfonlayer}{nodelayer}
		\node [style=none] (0) at (0, -0.75) {};
		\node [style=none] (1) at (0.5, 0.5) {};
		\node [style=none] (2) at (0.75, -2.325) {};
		\node [style=none] (4) at (0, -1.25) {$g$};
		\node [style=none] (5) at (1, 0.5) {};
		\node [style=none] (6) at (1.25, 2.175) {};
		\node [style=none] (7) at (0, -0.75) {};
		\node [style=none] (8) at (-0.25, -1.75) {};
		\node [style=none] (9) at (-0.5, -1.75) {};
		\node [style=none] (10) at (3, 2.5) {};
		\node [style=none] (11) at (1, 2.175) {};
		\node [style=none] (12) at (0.25, -0.75) {};
		\node [style=none] (13) at (-0.25, -1.75) {};
		\node [style=none] (14) at (0, -0.75) {};
		\node [style=none] (15) at (1.25, -0.5) {};
		\node [style=none] (16) at (-0.75, -2.5) {};
		\node [style=none] (17) at (3, -2.5) {};
		\node [style=none] (18) at (0.75, -2.2) {};
		\node [style=none] (19) at (0, 0.5) {};
		\node [style=none] (20) at (0, -3) {};
		\node [style=none] (21) at (1.25, 0.5) {};
		\node [style=none] (22) at (0.5, -0.5) {};
		\node [style=none] (23) at (0, -0.5) {};
		\node [style=none] (24) at (1.25, 2.3) {};
		\node [style=none] (26) at (1.5, 2) {};
		\node [style=none] (27) at (1.5, -2) {};
		\node [style=none] (28) at (1.25, -2.325) {};
		\node [style=none] (29) at (-0.5, -0.5) {};
		\node [style=none] (30) at (1.25, -2.2) {};
		\node [style=none] (31) at (-0.75, 2.5) {};
		\node [style=none] (32) at (-0.75, -2) {};
		\node [style=none] (33) at (0.25, -1.75) {};
		\node [style=none] (34) at (0.5, 0) {$f$};
		\node [style=none] (35) at (-0.5, 0.5) {};
		\node [style=none] (36) at (-0.75, 2) {};
		\node [style=none] (37) at (0.5, -1.75) {};
		\node [style=none] (38) at (0.75, 2.175) {};
		\node [style=none] (39) at (-0.5, -0.75) {};
		\node [style=none] (40) at (0.5, -0.75) {};
		\node [style=none] (41) at (0, -1.75) {};
		\node [style=none] (42) at (1, -2.2) {};
		\node [style=none] (43) at (0, -0.5) {};
		\node [style=none] (44) at (1, -0.5) {};
		\node [style=none] (45) at (0.75, 2.3) {};
		\node [style=none] (46) at (0, -1.75) {};
		\node [style=none] (47) at (0.5, 0.75) {};
		\node [style=none] (48) at (-0.25, 0.75) {};
		\node [style=none] (49) at (0, 0.75) {};
		\node [style=none] (50) at (-0.25, 1.75) {};
		\node [style=none] (51) at (0, 1.75) {};
		\node [style=none] (52) at (0, 0.5) {};
		\node [style=none] (53) at (-0.5, 0.75) {};
		\node [style=none] (54) at (0.5, 1.75) {};
		\node [style=none] (55) at (0, 0.75) {};
		\node [style=none] (56) at (0, 1.25) {$h$};
		\node [style=none] (57) at (0, 1.75) {};
		\node [style=none] (58) at (0.25, 1.75) {};
		\node [style=none] (59) at (-0.5, 1.75) {};
		\node [style=none] (60) at (0.25, 0.75) {};
		\node [style=none] (61) at (-0.25, 0.75) {};
		\node [style=none] (62) at (0, 3) {};
	\end{pgfonlayer}
	\begin{pgfonlayer}{edgelayer}
		\draw [style=boldedge] (35.center) to (29.center);
		\draw [style=boldedge] (15.center) to (29.center);
		\draw [style=boldedge] (15.center) to (21.center);
		\draw [style=boldedge] (21.center) to (35.center);
		\draw [style=boldedge] (5.center) to (11.center);
		\draw [style=boldedge] (42.center) to (44.center);
		\draw [style=boldedge] (45.center) to (24.center);
		\draw [style=boldedge] (6.center) to (38.center);
		\draw [style=boldedge] (30.center) to (18.center);
		\draw [style=boldedge] (2.center) to (28.center);
		\draw [style=boldedge] (14.center) to (43.center);
		\draw (17.center) to (10.center);
		\draw (10.center) to (31.center);
		\draw (31.center) to (36.center);
		\draw (36.center) to (26.center);
		\draw (26.center) to (27.center);
		\draw (27.center) to (32.center);
		\draw (32.center) to (16.center);
		\draw (16.center) to (17.center);
		\draw [style=boldedge] (39.center) to (9.center);
		\draw [style=boldedge] (37.center) to (9.center);
		\draw [style=boldedge] (37.center) to (40.center);
		\draw [style=boldedge] (40.center) to (39.center);
		\draw [style=boldedge] (41.center) to (20.center);
		\draw [style=boldedge] (51.center) to (62.center);
		\draw [style=boldedge] (59.center) to (53.center);
		\draw [style=boldedge] (47.center) to (53.center);
		\draw [style=boldedge] (47.center) to (54.center);
		\draw [style=boldedge] (54.center) to (59.center);
		\draw [style=boldedge] (55.center) to (52.center);
	\end{pgfonlayer}
\end{tikzpicture}
	\]
\end{lemma}
\begin{proof}
	For the LHS we get:
	\begin{equation*}
		\tikzfigscale{0.8}{naturalityCR_proof}
	\end{equation*}
	where $\tau$ satisfies:
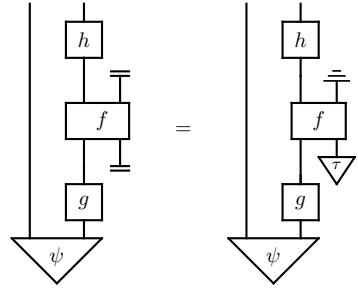
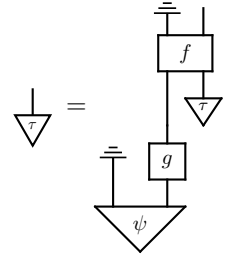
	\begin{equation*}
		\tikzfigscale{0.8}{tau} = \tikzfigscale{0.8}{naturalityCR_proof2}
	\end{equation*}
	For the RHS we get:
	\begin{equation*}
		\tikzfigscale{0.8}{naturalityCR_proofRHS}
	\end{equation*}
	where $\tau'$ satisfies the same condition as $\tau$:
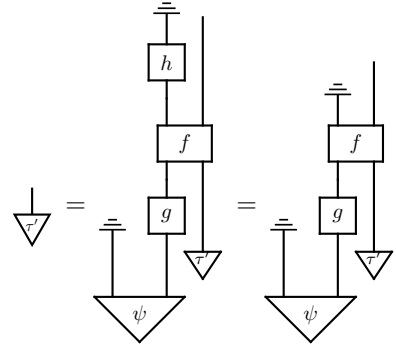
	\begin{equation*}
		\tikzfigscale{0.8}{tauprime} = \tikzfigscale{0.8}{naturalityCR_fixedpoint} = \tikzfigscale{0.8}{naturalityCR_fixedpoint2}
	\end{equation*}
	Hence we have that $\tau = \tau'$, and hence both sides are equal for all bipartite states $\psi$.
\end{proof}

\begin{lemma}
\label{lemma:deutsch_strength}
	The super-operator defined in the D-CTC model satisfies the strength property:
	\[
		\tikzfigscale{0.8}{superposing} \hspace{2mm}=\hspace{2mm}\tikzfigscale{0.8}{superposing2}
	\]
\end{lemma}
\begin{proof}

	Interestingly, the strength property is a consequence of the definition of the tensor product in $\DMix$ and of naturality in the CR region:
	\[
		\tikzfigscale{0.7}{superposing_proof} =  \tikzfigscale{0.7}{superposing_proof2} = \tikzfigscale{0.7}{superposing_proof3}
	\]
\end{proof}

\begin{lemma}
\label{lemma:deutsch_sliding}
	The super-operator defined in the D-CTC model satisfies the sliding property:
	\[
		\tikzfigscale{0.8}{sliding1} \hspace{2mm}=\hspace{2mm} \tikzfigscale{0.8}{sliding2}
	\]
\end{lemma}
\begin{proof}
	Out of the four properties defining time-travel super-operators, this is certainly the hardest (and most interesting) one to prove for the D-CTC model.

	Let $\mathcal{T}(\rho,\sigma)$ be the trace distance between two normalized state $\rho,\sigma$ on a $d$-dimensional Hilbert space $\mathcal{H}$:
	\[
		\mathcal{T}(\rho,\sigma)
		:=
		\dfrac{1}{2}\operatorname{Tr}[\sqrt{(\rho-\sigma)^\dagger(\rho-\sigma)}]
	\]
	The traced distance can also be expressed in terms of the family of eigenvalues $(\lambda_i)_{i=1}^d$ of the difference $\rho - \sigma$:
	\[
		\mathcal{T}(\rho,\sigma) = \dfrac{1}{2}\sum_{i}|\lambda_i|
	\]
	The trace distance endows the set of normalized states on $\mathcal{H}$ with a metric, and it is possible to show that CPTP maps are contractive with respect to this metric (see e.g. \cite{nielsen_chuang_book}, p. 406--407):
	\[
		\mathcal{T}(f(\rho),f(\sigma)) \leq \mathcal{T}(\rho,\sigma)
	\]
	Now let $f: \mathcal{A} \rightarrow \mathcal{B}$ and $g: \mathcal{B} \rightarrow \mathcal{A}$ be arbitrary CPTP maps, and write $P:=\suchthat{\rho}{(g \circ f)(\rho) = \rho}$ for the set of normalized states on $\mathcal{A}$ which are left fixed by $g \circ f$, and $Q:=\suchthat{\sigma}{(f \circ g)(\sigma) = \sigma} $ for the set of normalized states of $\mathcal{B}$ which are left fixed by $f \circ g$. 
	We wish to show that the $f$ and $g$ restrict to well-defined functions $f: P \rightarrow Q$ and $g: Q \rightarrow P$, and that these restrictions are mutual inverses, i.e. that $g \circ f = \id{P}$ and that $f \circ g = \id{Q}$. As a consequence, $P$ and $Q$ will be shown to be isometric via $f$ and $g$.
	Indeed, let $x \in P$ be a normalized state fixed by $g \circ f$:
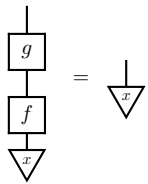
	 \[
	 	\tikzfigscale{0.8}{proof_fixedpoint}
	 \]
	 Applying $f$ to the both side of the equation we obtain:
	 \[
	 	\tikzfigscale{0.8}{proof_fixedpoint2}
	 \]
	 Similarly, $g$ applied to a normalized state $y \in Q$ left fixed by $f \circ g$ yields:
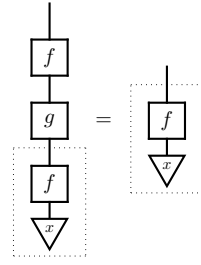
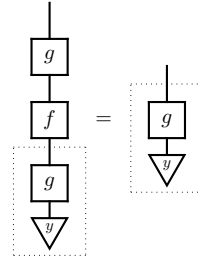
	 \[
	 	\tikzfigscale{0.8}{proof_fixedpoint3}
	 \]
	 Taking the last two equations together yields the desired result.

	Now we consider the geometric definition of the state of maximum entropy given in Appendix \ref{appendix:fixedPoint}, so that the fixed points of maximum entropy for $g \circ f$ and $f \circ g$ respectively can be written as:
	\[
		\tau_P = \int_P \rho \ d\rho
	\]
	and
	\[
		\tau_Q = \int_Q \sigma \ d\sigma
	\]
	By linearity and continuity of $f$ we have:
	\[
		f(\tau_P) = f\left(\int_P \rho \  d\rho \right) = \int_P f(\rho) \  d\rho
	\]
	By the fact that $f$ is an isometry when restricted to $P$ we have:
	\[
		\int_P f(\rho) \ d\rho = \int_Q \rho' \ d\rho' = \tau_Q
	\]
	We can conclude that:
	\[
		f(\tau_P) = \tau_Q
	\]
	Analogously, we can conclude that:
	\[
		g(\tau_Q) = \tau_P
	\]
	Hence $f$ maps the point of maximal entropy in $P$ to the point of maximal entropy in $Q$, and conversely $g$ maps the point of maximal entropy in $Q$ to the point of maximal entropy in $P$.

	We can now proceed with the final part of this proof. Consider the two CPTP maps $f$ and $g$ given by:
	\[
		\tikzfigscale{0.8}{slidingmaps}
	\]
	Let $\tau_P$ be the fixed point of maximal entropy for $g \circ f$, and $\tau_Q$ be the fixed point of maximal entropy for $f \circ g$. Then we have:
	\[
		\tikzfigscale{0.8}{slidingmaps_conclusion}
	\]
	This concludes our proof.
\end{proof}

\begin{lemma}
\label{lemma:deutsch_vanishing}
	The super-operator defined in the D-CTC model satisfies the vanishing property:
	\[
		\tikzfigscale{0.8}{vanishing1} \hspace{2mm}=\hspace{2mm} \tikzfigscale{0.8}{vanishing2}
	\]
\end{lemma}
\begin{proof}
	This is actually a special case of Lemma \ref{lemma:MixDMix_label}
\end{proof}

\setcounter{theorem:save}{\value{theorem_c}}

\setcounter{theorem_c}{\value{theorem:CVlocalSemantics}}
\begin{theorem}
	In the presence of a time-travel super-operator, morphisms defined by diagrams over CV-local framed causal graphs are well-defined. Conversely, any super-operator which yields well-defined diagrams over CV-local framed causal graphs (and respects the embedding of the original process theory) must satisfy the properties of a time-traveling super-operator.
\end{theorem}
\begin{proof}
Consider a time-travel super-operator $\Xi$. An arbitrary CV-local graph will have a finite number of vertices acting as interactions nodes $\{v_0, \ldots v_{N-1}\}$ and a finite number of cycles passing trough each interaction node. Without loss of generality, we will consider the case $N=1$; the general case is just a composition of individual interaction node cases. Regardless of how we decide to cut the cycles through $v_0$ open, strength and naturality allow us to gather all morphisms assigned to the vertices in the CR region--including the interaction nodes--in the same sub-diagram $\chi$. By the sliding property we can then move all the morphisms associated to CV region nodes (except the interaction node) to the another sub-diagram $CV$, at the output of sub-diagram $\chi$. This yields the following normal form:
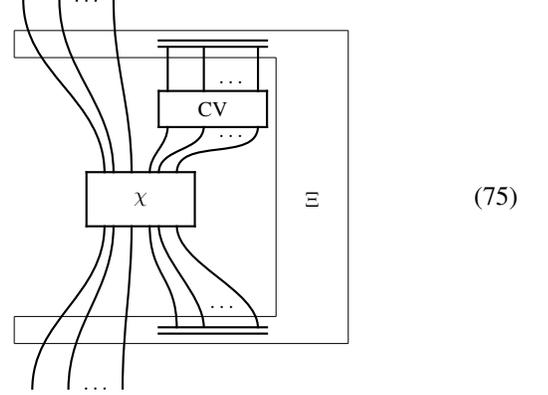
\begin{equation}
	\tikzfigscale{0.8}{normal_form}
\end{equation}
This is clearly independent of the way in which we cut the cycles, as the relative position between the boxes in the CV region is preserved. As a consequence, a time-travel super-operator defines the same morphism for every possible way of cutting the same CV-local graph open into a CR graph. \\

Conversely, if $\Xi$ yields well-defined diagrams over causal graphs, it can be easily shown that $\Xi$ is a time-travel super-operator. To begin with, we have already seen that Sliding is necessary to ensure well-definition, because there is no canonical way to cut cycles open in arbitrary CV-local graphs. Furthermore, Strength and Naturality can both be derived by looking at the same CV-local graph:
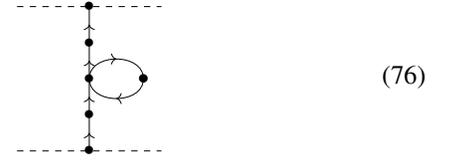
\begin{equation}
	\tikzfigscale{0.8}{well_definedness}
\end{equation}
Strength follows from the fact that the two following graphs need to represent the same morphism, for any arbitrary set of morphisms assigned to vertices:
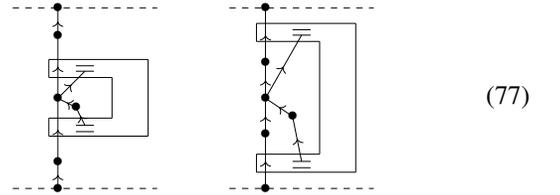
\begin{equation}
	\tikzfigscale{0.8}{well_definedness_strength}
\end{equation}
Naturality follows from the two possible ways one can cut open the cycles:
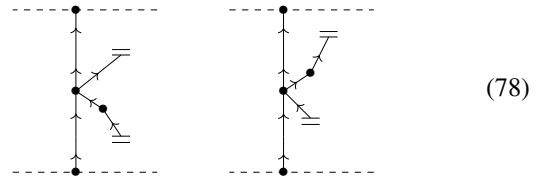
\begin{equation}
	\tikzfigscale{0.8}{well_definedness_naturality}
\end{equation}
Finally, Vanishing follows from the requirement that the super-operator respects the embedding of the original process theory.
\end{proof}

\setcounter{theorem_c}{\value{theorem:DeutschCannotDoAllCVgraphs}}
\begin{theorem}
	Consider a process theory with discarding maps and time travel. If the time-travel super-operator yields well-defined morphisms for diagrams over arbitrary framed causal graphs, then it must satisfy the yanking property.
\end{theorem}
\begin{proof}
Let $\rho$ be a normalised state. Strength and normalisation of $\rho$ imply that:
		\begin{equation}
		\label{eqn:equation_yanking}
		\tikzfigscale{0.8}{proof_yanking}
	\end{equation}
To show that the yanking property holds we can consider the following causal graph:
	\begin{equation}
		\tikzfigscale{0.8}{causal_graph}
	\end{equation}
Cutting the cycle open leads to the following two graphs:
	\begin{equation}
	\label{eqn:equivalence_graphs}
		\tikzfigscale{0.8}{simple_time_travel}
	\end{equation}
Assigning to each vertex the following morphisms:
	\begin{equation}
		f \colon = \tikzfigscale{0.8}{proof_unitality_f} \ \
		g \colon = \tikzfigscale{0.8}{proof_unitality_g}
	\end{equation}
Following from Equation (\ref{eqn:equation_yanking}), by the equivalence of graphs described in Equation (\ref{eqn:equivalence_graphs}) and strength:
			\begin{equation}
		\tikzfigscale{0.8}{proof_yanking2}
	\end{equation}
By sliding and vanishing we can conclude that:
			\begin{equation}
		\tikzfigscale{0.8}{proof_yanking3}
	\end{equation}
\end{proof}

\end{document}